\documentclass[12pt, draftclsnofoot, onecolumn]{IEEEtran}

\usepackage[T1]{fontenc}

\ifCLASSINFOpdf
  \usepackage[pdftex]{graphicx}
  \graphicspath{{}}
  \DeclareGraphicsExtensions{.pdf,.jpeg,.png}
\else
\fi

\usepackage[cmex10]{amsmath}
\usepackage{graphicx}
\usepackage{amsmath}
\usepackage{amssymb}
\usepackage{bm}
\usepackage{subfigure}
\usepackage{tabu}
\usepackage{cite}
\usepackage{indentfirst}
\usepackage{multirow}
\usepackage{hhline}
\usepackage{units}

\usepackage{epstopdf}
\usepackage{amsxtra}
\usepackage{amstext}
\usepackage{amssymb}
\usepackage{amsthm}
\usepackage{latexsym}
\usepackage{dsfont} 
\usepackage{color}
\usepackage{xcolor}
\usepackage{setspace}
\usepackage{bbm} 

\newcommand{\blue}[1]{{\color{blue}{#1}}} 
\newcommand{\red}[1]{{\color{red}{#1}}} 

\usepackage{algorithm}
\usepackage[noend]{algpseudocode}
\algnewcommand{\LineComment}[1]{\State \(\triangleright\) #1}

\makeatletter
\def\BState{\State\hskip-\ALG@thistlm}
\makeatother

\usepackage{graphicx}

\theoremstyle{remark}
\newtheorem{remark}{Remark}

\makeatletter
\let\NAT@parse\undefined
\makeatother
\ifCLASSINFOpdf
\else
\fi

\definecolor{darkgreen}{rgb}{0,0.75,0}
\usepackage{hyperref}
\hypersetup{
  pdfauthor = {Mohammad M. Mansour},
  pdfcreator = {Mohammad M. Mansour},
  pdfnewwindow,
  pdffitwindow=true,
  colorlinks=true,
  linkcolor=blue,       
  citecolor=darkgreen,      
  filecolor=red,        
  urlcolor=cyan         
}

\newcommand{\abs}[1]{\left|{#1}\right|}

\newcommand{\expec}[1]{{\mathbb{E}\!\left[{#1}\right]}}
\newcommand{\prob}[1]{{\mathbb{P}\!\left[{#1}\right]}}

\newcommand{\nth}[1]{{#1}{\text{th}}}

\newcommand{\covprobCC}[0]{\bar{\mathcal{P}}_{\mathrm{cov},\mathrm{C}}}
\newcommand{\covprobCCn}[0]{{\mathcal{P}}_{\mathrm{cov},\mathrm{C}}}
\newcommand{\covprobCn}[1]{\covprobCCn\!\left({#1}\right)}
\newcommand{\covprobC}[1]{\covprobCC\!\left({#1}\right)}

\newcommand{\covprobDD}[0]{\bar{\mathcal{P}}_{\mathrm{cov},\mathrm{D}}}
\newcommand{\covprobD}[1]{\covprobDD\!\left({#1}\right)}

\newcommand{\covprobDDn}[0]{{\mathcal{P}}_{\mathrm{cov},\mathrm{D}}}
\newcommand{\covprobDn}[1]{\covprobDDn\!\left({#1}\right)}

\newcommand{\txprob}{\mathcal{P}_{\mathrm{tx}}}

\newcommand{\RC}{R_{\mathrm{C}}}
\newcommand{\RD}{R_{\mathrm{D}}}

\theoremstyle{definition}

\usepackage{amsthm}

\newtheorem{theorem}{Theorem}
\newtheorem{corollary}{Corollary}
\newtheorem{lemma}{Lemma}
\newtheorem{proposition}{\textit{Proposition}}

\begin{document}

\title{Power Control and Channel Allocation for D2D Underlaid Cellular Networks}

\author{Asmaa~Abdallah,~\IEEEmembership{Student Member,~IEEE,}
        Mohammad~M.~Mansour,~\IEEEmembership{Senior Member,~IEEE,}
        and~Ali~Chehab,~\IEEEmembership{Senior Member,~IEEE}
\thanks{The authors are with the Department
of Electrical and Computer Engineering, American University of Beirut, Lebanon. E-mail: \texttt{\{awa18,mmansour,chehab\}@aub.edu.lb.}}
}

\maketitle
\vspace{-0.6in}
%
\begin{abstract}\vspace{-0.1in}
Device-to-Device (D2D) communications underlaying cellular networks is a viable network technology that can potentially increase spectral utilization and improve power efficiency for proximity-based wireless applications and services. However, a major challenge in such deployment scenarios is the interference caused by D2D links when sharing the same resources with cellular users. In this work, we propose a channel allocation (CA) scheme together with a set of three power control (PC) schemes to mitigate interference in a D2D underlaid cellular system modeled as a random network using the mathematical tool of stochastic geometry. The novel aspect of the proposed CA scheme is that it enables D2D links to share resources with \emph{multiple} cellular users as opposed to one as previously considered in the literature. Moreover, the accompanying distributed PC schemes further manage interference during link establishment and maintenance. The first two PC schemes compensate for large-scale path-loss effects and maximize the D2D sum rate by employing distance-dependent path-loss parameters of the D2D link and the base station, including an error estimation margin. The third scheme is an adaptive PC scheme based on a variable target signal-to-interference-plus-noise ratio, which limits the interference caused by D2D users and provides sufficient coverage probability for cellular users. Closed-form expressions for the coverage probability of cellular links, D2D links, and sum rate of D2D links are derived in terms of the allocated power, density of D2D links, and path-loss exponent. The impact of these key system parameters on network performance is analyzed and compared with previous work. Simulation results demonstrate an enhancement in cellular and D2D coverage probabilities, and an increase in spectral and power efficiency.
\end{abstract}

\begin{IEEEkeywords}
Device-to-device communications, Poisson point process, power control, resource allocation, stochastic geometry.
\end{IEEEkeywords}

%
\section{Introduction}\label{sec:Power Control Types}
The main motivation behind using Device-to-Device (D2D) communication underlaying cellular systems is to enable communication between devices in close vicinity with low latency and low energy consumption, and potentially to offload a telecommunication network from handling local traffic \cite{doppler2009device,corson2010toward, fodor2012design,3gpp803,lin2014overview}. D2D is a promising approach to support proximity-based services such as social networking and file sharing~\cite{3gpp803}. When the devices are in close vicinity, D2D communication improves the spectral and energy efficiency of cellular networks~\cite{lin2014overview}.

Despite the benefits of D2D communications in underlay mode, interference management and energy efficiency have become fundamental requirements \cite{mach2015band} in keeping the interference caused by the D2D users under control, while simultaneously extending the battery lifetime of the User Equipment (UE). For instance, cellular links experience cross-tier interference from D2D transmissions, whereas D2D links not only deal with the inter-D2D interference, but also with cross-tier interference from cellular transmissions. Therefore, power control (PC) and channel allocation (CA) have become necessary for managing interference levels, protecting the cellular UEs (CUEs), and providing energy-efficient communications.

Power control and channel allocation schemes have been presented in the literature as strategies to mitigate interference in wireless networks \cite{fodor2013comparative,yates1997soft,de2015power,de2014power,de2015uplink,ali2015mode,lee2015power,memmi2016power,banagar2016power,islam2016two,Resource_VTC,liu2016optimizing,liu2016optimizing2,huang2015mode, wang2016hypergraph,tang2016mixed,wang2016pairing}. In \cite{fodor2013comparative}, open loop and closed loop PC schemes (OLPC, CLPC), used in LTE \cite{3gpp213}, are compared with an optimization based approach aimed at increasing spectrum usage efficiency and reducing total power consumption. However, such schemes require a large number of iterations to converge.

In~\cite{yates1997soft,de2015power,de2014power,de2015uplink}, a power allocation scheme is presented based on a ``soft dropping'' PC algorithm, in which the transmit power meets a variable target signal-to-interference-plus-noise ratio (SINR). However, the system considered is not random, and the D2D users in \cite{de2015power,de2014power,de2015uplink} are confined within a hotspot in a cellular region.

In \cite{ali2015mode}, a D2D ``mode'' is selected in a device based on its proximity to other devices and to its distance to the eNB. However, the inaccuracy of distance derivation is a key aspect that is not addressed in~\cite{ali2015mode}. In \cite{islam2016two}, a two-phase auction-based algorithm is used to share uplink spectrum. The authors assume that all the channel information is calculated at the eNB and broadcasted to users in a timely manner, which will cause an excessive signaling overhead. In \cite{Resource_VTC}, a heuristic delay-tolerant resource allocation is presented for D2D underlying cellular networks; however, power control is ignored since D2D users always transmit at maximum power.

In the above schemes, power control and channel allocation methods \cite{fodor2013comparative,yates1997soft,de2015power,de2014power,de2015uplink,islam2016two,Resource_VTC,tang2016mixed,wang2016pairing} are developed and evaluated assuming deterministic D2D link deployment scenarios. On the other hand, PC in \cite{lee2015power} is presented for unicast D2D communications by modeling a random network for a D2D underlaid cellular system,  using stochastic geometry. D2D users are distributed using a (2-dimensional) spatial Poison point process (PPP) with density $\lambda$. Stochastic geometry is a useful mathematical tool to model irregular spatial structures of D2D locations, and to quantify analytically the interference in D2D underlaid cellular networks using the Laplace transform~\cite{schilcher2016interference,stoyan1995stochastic,baccelli2009stochastic}. Two PC schemes are developed in~\cite{lee2015power}; a centralized PC and a simple distributed on-off PC scheme. The former requires global channel state information (CSI) possibly at a centralized controller, which may incur high CSI feedback overhead, whereas the latter is based on a decision set and requires only direct link information. However, the authors assume a fixed distance between the D2D pairs, and that the D2D devices for the distributed case operate at maximum power leading to severe co-channel interference. Moreover, the distributed PC scheme of~\cite{lee2015power} does not guarantee reliable cellular links, especially at high SINR targets. In~\cite{memmi2016power}, similar PC algorithms to~\cite{lee2015power} are presented but with channel uncertainty considered; the results in~\cite{lee2015power} are regarded as ideal best-case scenarios with perfect channel knowledge. In~\cite{liu2016optimizing,liu2016optimizing2}, a framework based on stochastic geometry to analyze the coverage probability and average rate with different channel allocations in a D2D overlaid cellular systems is presented.

In~\cite{huang2015mode}, PC and resource allocation schemes are considered; however, the interference between D2D pairs is ignored.  In~\cite{wang2016hypergraph}, a transmission cost minimization problem using hypergraph model is investigated based on a content encoding strategy to download a new content item or repair a lost content item in D2D-based distributed storage systems. Moreover,~\cite{wang2016hypergraph} considers the one-to-one matching case, in which only one D2D link shares resources with only one uplink cellular user. In~\cite{tang2016mixed,wang2016pairing}, resource allocation is considered where one D2D link shares resources with only one cellular user in the underlay case. Obviously, these schemes in~\cite{huang2015mode, wang2016hypergraph,tang2016mixed,wang2016pairing} are not spectrally efficient because D2D pairs are restricted to use different resource allocations. In~\cite{zhang2012random,jindal2008fractional}, power control is studied in random ad hoc networks without taking into consideration the underlaid cellular network.

In this paper, we propose power control methods along with channel allocation and analyze their performance assuming a random D2D underlaid cellular network model. A main shortcoming in most papers in the literature is that unrealistic assumptions are considered. For instance, in~\cite{lee2015power,memmi2016power} the authors rely on deterministic values such as fixed distance between the D2D transmitter and receiver, fixed transmission power, and fixed SINR targets and they only consider one cellular user sharing the resources with the D2D links. These deterministic assumptions simplify the derivation of the analytical models, but are in many cases unrealistic. In our work, we study a general scenario by randomly modeling the distance between the D2D pairs, assigning different transmission power to D2D links, varying the SINR targets, and consider multiple cellular users sharing the resources with the D2D links. Therefore, the presented analytical expressions in this paper give more insight into the performance of a D2D underlaid cellular system in a rather more realistic approach.

%
\emph{Contributions}: The main contributions of this work are the following:
\begin{enumerate}
\item A new channel allocation scheme is proposed based on how far the D2D users are from the cellular users. It enables D2D links to share resources with multiple cellular users as opposed to one as previously considered in the literature. It also decreases the density of active D2D users sharing the same resources, thus the interference generated by the D2D users is decreased, which in turn enhances the cellular as well as the D2D coverage probabilities.

\item Analytical expressions for the coverage probability for cellular and D2D links are derived taking into account \emph{varying} distances between the pairs of devices, in contrast to \cite{lee2015power,memmi2016power}. Therefore, the random variables that model distances and allocated power will significantly add to the complexity of the equations derived in \cite{lee2015power,memmi2016power}; however, the randomness of the D2D underlaid system is efficiently captured, and accurate insights of the performance aspects of the D2D system are provided.

\item Two distributed power control algorithms are proposed for \emph{link establishment}. One scheme maximizes the sum rate of the D2D links, while the other minimizes the interference level at the eNB. Both schemes depend on a distance-based path-loss parameter between the D2D transmitter, D2D receiver and the eNB. In addition, the inaccuracy of distance estimation is handled by incorporating an estimation error margin. A closed-form expression of various moments of the power allocated to the D2D links is derived. Moreover, an analytical expression of the sum-rate of D2D links is derived to determine the optimal D2D transmission probability that maximizes this sum rate.

\item A distributed adaptive power control scheme (soft dropping distance-based PC) is proposed for \emph{link maintenance}.  This PC scheme adapts to channel changes in a more realistic manner. Furthermore, this dynamic approach maintains the link quality over time by softly dropping the target SINR as the distance between the D2D pairs changes, and thus the power transmitted is adjusted to meet this variable SINR. Hence, this scheme limits the interference caused by the D2D users while varying the target SINR for the D2D links.
\end{enumerate}

The rest of the paper is organized as follows. The system model for a D2D underlaid cellular network is described in Section~\ref{s:system_model}. In Section~\ref{s:Binary_CA}, the proposed channel allocation is introduced. In Section~\ref{s:coverage_prob}, analytical expressions for the coverage probabilities are derived. In Section~\ref{s:pc_schemes}, the proposed PC schemes are presented. Case studies with numerical results are simulated and analyzed based on the proposed schemes in Section~\ref{sec:simulation}. Section~\ref{sec:conclusion} concludes the paper.

\begin{figure}[t]
    \centering
    \includegraphics[width=0.5\textwidth]{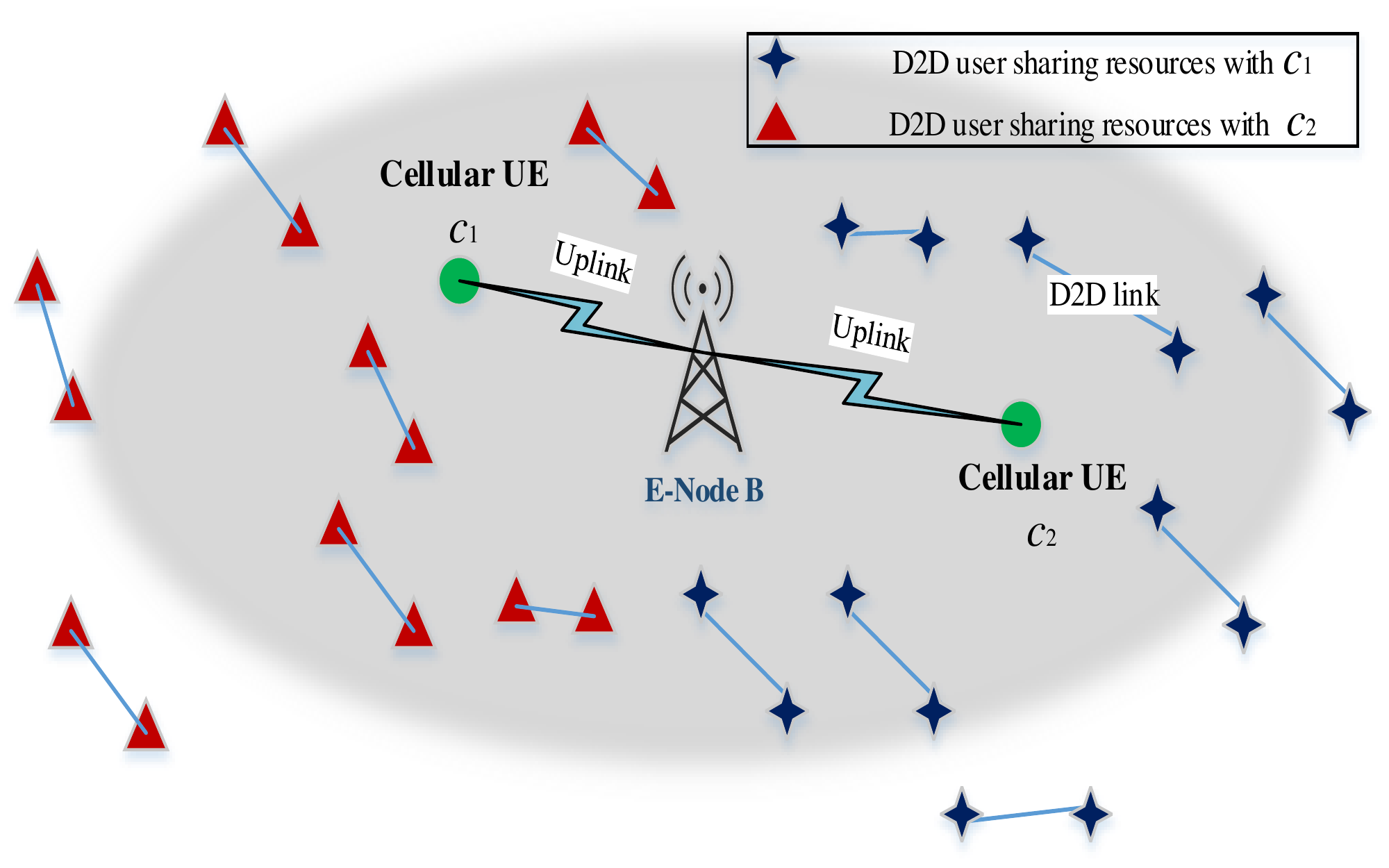}\vspace{-0.25in}
    \caption{A single-cell D2D underlaid cellular network. Two cellular users $c_1$ and $c_2$ establish a link with the eNB while several active D2D links are established in a disk centered at the eNB with radius $\RC$. For the case $m=2$, a subset of active D2D links share resources with cellular UE $c_1$ (\blue{$\bigstar$}), while other D2D links share resources with $c_2$ (\red{$\blacktriangle$}). }
\label{fig:sys_model1}
\end{figure}

\begin{figure}[t]
    \centering
    \includegraphics[width=0.5\textwidth]{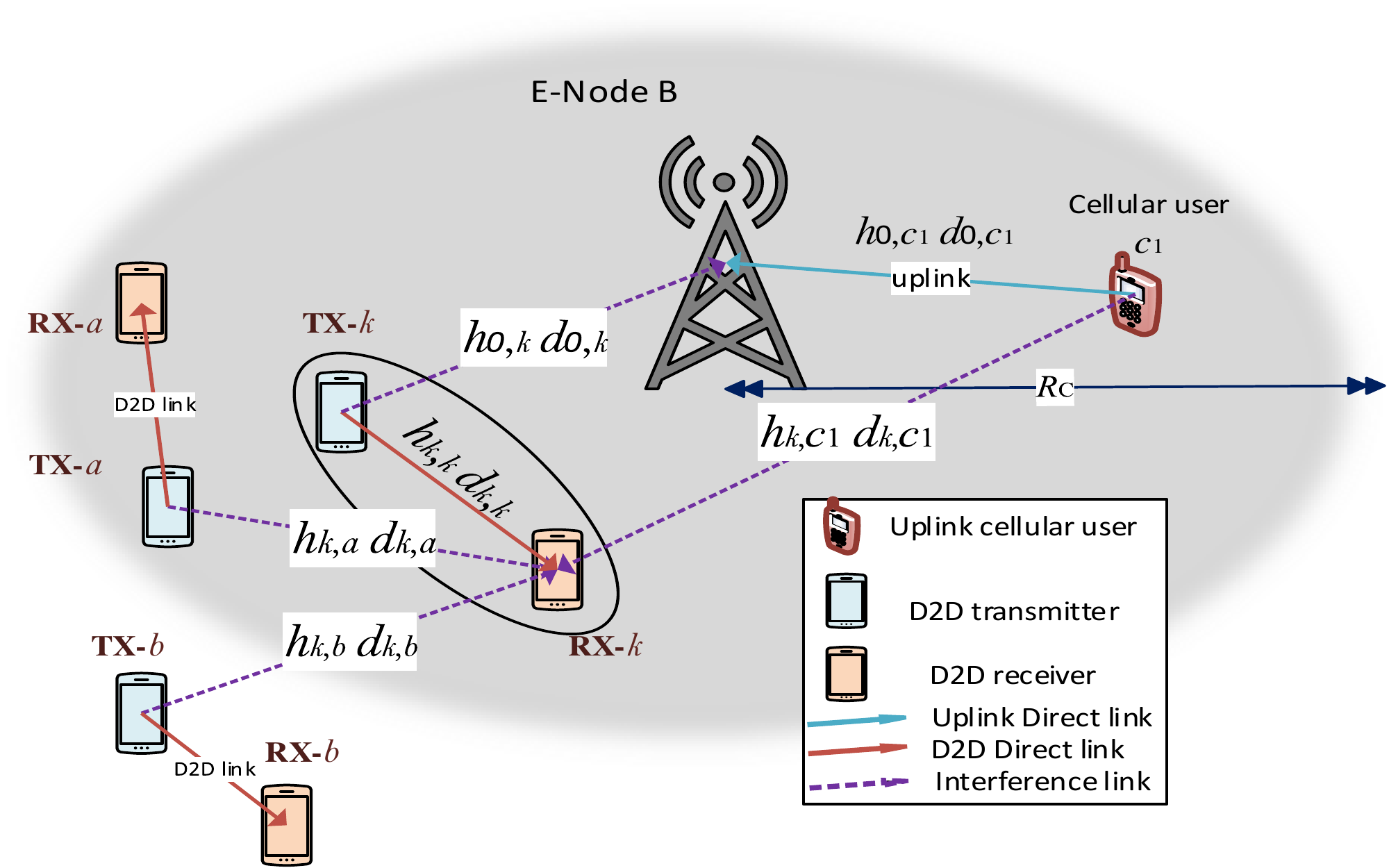}\vspace{-0.25in}
    \caption{The system model shows the channel model for one of the cellular users and a subset of active D2D links that share resources with $c_1$. The active D2D links outside the cell are considered as out-of-cell D2D interference, whereas out-of-cell interference from cellular users belonging to cross-tier cells is ignored. }
\label{fig:sys_model}
\end{figure}

%
\section{System Model}\label{s:system_model}
In this section, the system model and the corresponding network parameters are presented. As shown in Fig.~\ref{fig:sys_model1}, we study a D2D underlaid cellular network in which a pool of $K$ active D2D users is divided into $M$ groups such that each group shares distinct resources with one of $M$ cellular users, as opposed to the assumption taken in \cite{lee2015power,memmi2016power} where all the $K$ D2D users share the same resource with \textit{one} cellular user. The eNB coverage region is modeled as a circular disk $\mathcal{C}$ with radius $\RC$ and centered at the eNB. We assume that two cellular users are uniformly distributed in this disk, while the D2D transmitters are distributed in the whole $\mathsf{R}^2$ plane by the homogeneous PPP $\Phi$ with density $\lambda$, where $\prob{\Phi=n} = \exp\left(-\lambda\right)\tfrac{\lambda^n}{n!}$. The PPP assumption corresponds to having the expected number of nodes per unit area equal to $\lambda$, and the nodes being uniformly distributed in the area of interest. Hence, the number of D2D transmitters in $\mathcal{C}$ is a Poisson random variable $K$ with mean $\expec{K}=\lambda\pi R_{\mathrm{C}}^2$. In addition, the associated D2D receiver  is uniformly distributed in a disk centered at its transmitter with radius $R_{\mathrm{D}}$.

We consider a particular realization of the PPP $\Phi$ and a transmission time interval (TTI) $t$ to describe the system model. In the following, we use subscript 0 to refer to the uplink signal received by the eNB, $c_m$ to refer to the $\nth{m}$ transmitting cellular user, and $k\neq 0$ to refer to the $\nth{k}$ D2D user. Denote by $s_{0,c_m}^{(t)}$ the signal transmitted by the $\nth{m}$ cellular user in the uplink, and by $s_{k,k}^{(t)}$ the signal transmitted by the $\nth{k}$ D2D \emph{transmitter} to its $\nth{k}$ D2D \emph{receiver}, during the TTI $t$. We assume distance-independent Rayleigh fading channel models between the eNB and the UEs, between the eNB and the D2D users, and between the D2D users themselves.
Let $h_{0,c_m}^{(t)}$ denote the uplink channel gain between the $\nth{m}$ cellular user and eNB, $h_{k,k}^{(t)}$ the direct link channel gain between the $\nth{k}$ D2D \emph{transmitter} (TX) and corresponding $\nth{k}$  D2D \emph{receiver} (RX), $h_{0,k}^{(t)}$ the channel gain of the interfering link from the $\nth{k}$  D2D TX to the eNB, $h_{k,c_m}^{(t)}$ the channel gain of the interfering link from the $\nth{m}$ cellular UE to the $\nth{k}$ D2D RX, and $h_{k,l}^{(t)}$ the lateral channel gain of the interfering link from the $\nth{l}$ D2D TX to the $\nth{k}$ D2D RX. Random variables $n_0$ and $n_k$ denote additive noise at the eNB and the $\nth{k}$ D2D RX, and are distributed as $C\mathcal{N}(0,\sigma^2)$, where $\sigma^2$ is the noise variance. We also assume a distance-dependent path-loss model, i.e., a factor of the form $d_{k,l}^{-\alpha}$ that modulates the channel gains, where $d_{k,l}$ represents the distance between the $\nth{l}$ TX and the $\nth{k}$ RX, with $\alpha$ being the path-loss exponent.

Moreover, we assume that each cellular user and a subset $K' < K$  of the D2D transmitters share the same uplink physical resource block (PRB) during the same TTI ($t$) as depicted in Fig.~\ref{fig:sys_model}. Furthermore, we assume that the channel coherence bandwidth is larger than the bandwidth of a PRB, leading to a flat fading channel over each PRB.  Therefore, the received signals $y_{k,k}^{(t)}$ at the $\nth{k}$ D2D \emph{receiver}, and $y_{0,c_m}^{(t)}$ at the eNB can be expressed as
\begin{IEEEeqnarray}{rClrCl}
    y_{k,k}^{(t)} &=& h_{k,k}^{(t)}{d^{(t)}_{k,k}}^{-\alpha/2}s_{k,k}^{(t)}
                + h_{k,c_m}^{(t)}{d^{(t)}_{k,c_m}}^{-\alpha/2} s_{0,c_m}^{(t)}
                + \sum_{l=1,l\neq k}^{K'}h_{k,l}^{(t)}{d^{(t)}_{k,l}}^{-\alpha/2} s_{k,l}^{(t)}+n_k^{(t)},\label{equ1}\\
    y_{0,c_m}^{(t)}&=&h_{0,c_m}^{(t)}{d^{(t)}_{0,c_m}}^{-\alpha/2} s_{0,c_m}^{(t)}+\sum_{k=1}^{K'}h_{0,k}^{(t)}{d^{(t)}_{0,k}}^{-\alpha/2} s_{k,k}^{(t)}+ n_0^{(t)}.\label{equ2}
\end{IEEEeqnarray}

The transmit powers $p_0$ and $p_k$ are conditioned to meet certain peak power constraints, i.e. $p_0 \triangleq \abs{s_{0,c_m}}^2 \leq P_{\max,\mathrm{C}}$ and $p_k \triangleq |s_{k,k}|^2 \leq P_{\max,\mathrm{D}}$ for all links. The channel gains are estimated at each D2D receiver using the reference signal received power (RSRP), and are fed back to the corresponding D2D transmitter. In addition, it is worth noting that $\expec{K}$ represents the average number of D2D links (or transmitters) before channel allocation, whereas $\expec{K'}$ represents the number of D2D links (or transmitters)  sharing resources with $c_m$.

The SINR of any typical link is defined as $\text{SINR}\triangleq\dfrac{W}{I+N}$, where $W$ represents the power of the intended transmitted signal, $I$ represents the power of the interfering signals, and $N$ denotes the noise power. Therefore, the SINR at the eNB and D2D receiver $k$ can be written as
\begin{align}
    \text{SINR}_0(K',\textbf{p}) &=
        \dfrac{p_0 |h_{0,c_m}|^2d_{0,c_m}^{-\alpha} }{\sum_{k=1}^{K'} p_k|h_{0,k}|^2d_{0,k}^{-\alpha} \: +\: \sigma^2 },\label{equ4}\\
    \text{SINR}_k(K',\textbf{p}) &=
        \dfrac{p_k|h_{k,k}|^2d_{k,k}^{-\alpha} }{\sum_{i\neq 0,k}^{K'}p_i|h_{k,i}|^2d_{k,i}^{-\alpha} \: +\:  p_0|h_{k,c_m}|^2d_{k,c_m}^{-\alpha} \:  +\: \sigma^2 }, ~~k>0\label{equ:SINR_k}
\end{align}
where $\textbf{p}=[p_0,p_1,\cdots,p_k]^T$ represents the transmit power profile vector, with $p_i$ being the transmit power of the $\nth{i}$ UE transmitter, and $K'$ is the number of D2D transmitters. The super-subscript $(t)$ is suppressed for simplicity.

The proposed system model ignores the out-of-cell interference transmission from other uplink users from cross-tier cells. However, the density of the D2D links is a network parameter that captures the expected interference on cellular and D2D links. Moreover, when the density of the D2D links is high, the proposed system is able to capture the effect of the dominant interferer for both cellular (uplink) and D2D links, since there is a high probability that the nearest D2D interferer would become the dominant interference of a D2D link and that of the cellular link. Furthermore, when this network parameter is high, it can provide an upper bound on the performance of a D2D underlaid cellular network with out-of-cell interference. In addition, one can note that the radius of the disk $\RC$ is large enough to encompass all the D2D pairs, since the dominant interference is generated from the nearest D2D interferers.

Based on the above defined SINRs, we use the coverage probability and achievable sum rate as metrics to evaluate system performance. Precisely, the proposed CA and PC algorithms aim to maximize those quantities while maintaining a minimum level of Quality-of-Service (SINR threshold $\beta$). The coverage probabilities of both the cellular link and D2D links are derived in this work. The cellular coverage probability $\covprobC{\beta_0}$ is defined as
\begin{equation}\label{eq:cov_probC_def}
\covprobC{\beta_0}=\expec{\covprobCn{\textbf{p},\beta_0}}=\expec{\mathbb{P}(\text{SINR}_0 (K',\textbf{p})\geq\beta_0)},
\end{equation}
where $\beta_0$ denotes the minimum SINR value for reliable uplink connection. Similarly, the D2D coverage probability $\covprobD{\beta_k}$ is defined as
\begin{equation}\label{eq:cov_probD_def}
\covprobD{\beta_k}=\expec{\covprobDn{\textbf{p,}\beta_k}}=\expec{\mathbb{P}(\text{SINR}_k(K',\textbf{p}) \geq\beta_k)},
\end{equation}
where $\beta_k$ denotes the minimum SINR value for a reliable D2D link connection. In addition, the ergodic sum rate of D2D links is defined as
\begin{equation}\label{equ:rate}
\mathcal{R}_s^{(D)} =\mathbb{E} \left[\sum_{k=1}^{K'} \text{log}_2\left (1 + \text{SINR}_k(K',\textbf{p}) \right)\right].
\end{equation}
The main system parameters are summarized in Table~\ref{table:parameters}.
\begin{table}[t!]
\caption{System Parameters}
\centering
\label{table:parameters}
\begin{tabular}{l|l}\hline
    Cell radius & $R_{\mathrm{C}}$\\
    PPP of all D2D users in the cell & $\Phi$\\
    PPP of all D2D users in the cell after channel allocation & $\Phi'$\\
    Density of D2D links (D2D$/m^2$) & $\lambda$\\
    Channel gain from the cellular UE $c_m$  to eNB & $h_{0,c_m}$ \\
    Channel gain from D2D TX $k$ to D2D RX $k$ & $h_{k,k}$ \\
    Channel gain from D2D TX $k$ to eNB & $h_{0,k}$ \\
    Channel gain from the cellular UE $c_m$ to D2D RX $k$ & $h_{k,c_m}$ \\
    Channel gain from D2D TX $l$ to D2D RX $k$ & $h_{k,l}$ \\
    Distribution of channel fading ($h_{x,y}$) & Rayleigh fading $|h_{x,y}|^2\sim \exp{(1)}$\\
    Distance between D2D links ($d_{k,k}) $& Uniformly distributed \\
    Distance between uplink user and eNB ($d_{0,c_m}$) & Uniformly distributed \\
    Distance between D2D TX $k$ and eNB ($d_{0,k}$) & Uniformly distributed \\
    Expectation of an event  & $\expec{\cdot}$ \\
    Probability of an event & $\prob{\cdot}$ \\
    Laplace transform of a variable $X$ & $\mathcal{L}_X$ \\
    Coverage probability of link $L$ & ${\mathcal{P}}_{\text{cov},\mathrm{L}}$ \\
    Transmit probability &  $\txprob$\\
    Ergodic sum rate of D2D links & $\mathcal{R}_s^{(D)}$\\
    Maximum transmit power for cellular user & $P_{\max,\mathrm{C}}$\\
    Maximum transmit power for D2D user & $P_{\max,\mathrm{D}}$\\
    Receiver sensitivity (dBm) & $\rho_\mathrm{rx}$\\
    Cumulative distribution function (cdf) of variable $X$ & $F_X(\cdot)$\\
    Probability density function (pdf) of variable $X$ & $f_X(\cdot)$\\\hline
\end{tabular}
\end{table}

%
\section{Proposed Channel Allocation Scheme}\label{s:Binary_CA}
In this section, we propose a channel allocation scheme that enables active D2D users to share the same resource blocks used by $M\!>\!1$ (two or more) CUEs. Its main objective is to decrease the density of  D2D users sharing the same resource with a particular cellular user by dividing all active D2D pairs into $M$ groups, such that each group shares resources with one of $M$ distinct CUEs. By this, we further extend the system model in \cite{lee2015power,memmi2016power}, where only the case $M\!=\!1$ is considered and all the active D2D pairs are assumed to share resources with only one CUE. However, we include in our study a new resource allocation scheme for $M>1$.

Initially, mode selection determines whether a D2D pair can transmit in D2D mode or in cellular mode, and time and/or frequency resources are allocated accordingly. For simplicity, we study the case of two CUEs ($M=2$); the same approach is also generalized for any $M$. Using the independent thinning property \cite{andrews2016primer}, we independently assign random binary marks $\{1,2\}$ to the subset of active D2D users that can share resources with cellular users $c_1$ and $c_2$, respectively. The assignment is based on the following criterion: when the distance between the cellular UE $c_1$ and the $\nth{k}$ D2D RX is greater than the distance between cellular $c_2$ and the $\nth{k}$ D2D RX ($d_{k,c_1} > d_{k,c_2}$), the $\nth{k}$ D2D TX at instant $(t)$ will be assigned the value $\{ 1\}$; otherwise the D2D TX will be assigned $\{2\}$. Consequently, all D2D users assigned with value $\{1\}$ will share the same resources as $c_1$, while the rest will share the same resources as $c_2$. Therefore, sharing resources with the farthest cellular user reduces the interference at the eNB by decreasing the density of the D2D TXs sharing the same resources, and reduces the interference generated from the cellular user at the D2D RXs.

\begin{remark} Independent thinning of a PPP alters the density of the point process. If we independently assign random binary marks $\{1,2\}$ with $\mathbb{P}[Q_k\!=\!1]\!=\!q$ and $\mathbb{P}[Q_k\!=\!2]\!=\!1\!-\!q$ to each point in a PPP and collect all the points which are marked as ${1}$, the new point process will be a PPP $\Phi'$ but with density $q\lambda$, while the remaining points marked as ${2}$ will have a PPP $\Phi''$ with density $(1-q)\lambda$. In our case, the arrival of D2D users such that $d_{k,c_1}\!>\!d_{k,c_2}$ is independent of the arrival of another pair of D2D users such that $d_{k,c_2}\!>\!d_{k,c_1}$. Hence the thinning property applies.
\end{remark}

\begin{lemma}\label{lemma1} Using the above remark, half of the active D2D users will share the same resources with one of the cellular users and the other half will share them with the other cellular user.
\end{lemma}
\begin{proof}
The proof relies on the pdf of the distance between two uniformly distributed points, which is given by~\cite{moltchanov2012distance}:
\begin{equation}\label{equdkc}
f_{d_{k,c_m}}(r)\!=\!{2r\over R_\text{C}^{2}}\!\left(\!{2\over\pi}\cos^{-1}\!\left({r\over 2R_\text{C}} \right)\!\!-\!{r\over\pi R_\text{C}}\!\sqrt{1\!-\!{r^{2}\over 4R_\text{C}^{2}}}\right),\ \ 0\!\leq\! r\!\leq\! 2R_\text{C}.
\end{equation}
Using~\eqref{equdkc}, we have $\mathbb{P}[Q_k=1]=q={1\over 2}$. See detailed proof in Appendix~\ref{proof:lemma1}. A similar approach can be applied for a more general case of $M$ CUEs. A D2D UE shares resources with the CUE that is furthest away from it. For instance, if resources are shared with $M$ CUEs, then after $M-1$ comparisons, $\mathbb{P}[{d_{k,c_m}} \geq \max_{n\neq m} \{ {d_{k,c_n}},\cdots ,{d_{k,c_M^{}}}\}]=\mathbb{P}[Q_k=\text{value assigned to } c_m]=q_k=\tfrac{1}{M}$, where $\sum_i^M{q_i}=1$.
\end{proof}

We show in Section \ref{s:coverage_prob} that the coverage probabilities for cellular and D2D links depend on the density of the D2D users sharing the same resource. With the proposed CA scheme, the density of the D2D users is decreased by a factor of $q<1$ to be $q\lambda$. Therefore, the interference at the eNB is further reduced compared to the scenario considered in~\cite{lee2015power}, because here a smaller number of D2D users ($\mathbb{E}[K']=\tfrac{\mathbb{E}[K]}{M}=\tfrac{1}{M} \lambda \pi \RC^2$) share the same resources with the each CUE.

It should be noted that sharing resources with more than one CUE increases the coverage probability, which is intuitive as the interference caused by the D2D links is reduced. However, upon increasing $M$ (implying decreasing $K'$),  the spectral efficiency of the system will decrease according to ~(\ref{equ:rate}), and hence we would lose one of the main advantages of D2D communications that is increasing the spectral efficiency of the cellular system. Therefore, a trade-off exists between enhancing the link coverage probability and increasing the system throughput.

In addition, the complexity of the proposed channel allocation is $O(MK)$ where $K$ is the total number of the D2D links that will share resources with $(M>1)$ uplink cellular users. This is due to the fact that the base station will compute, for each D2D link, the distance $(d_{k,c_m} )$ from the $\nth{k}$ D2D receiver to the all $M$ cellular users (where $0< k\leq K$ and $1< m\leq M$). Therefore, the base station computes a total of $MK$ distance parameters $(d_{k,c_m} )$ to perform the comparisons (${d_{k,c_m}} \geq \max_{n\neq m} \{ {d_{k,c_n}},\cdots ,{d_{k,c_M^{}}}\}$) as discussed above.

%
\section{Analysis of Coverage Probability}\label{s:coverage_prob}

%
In this section, the cellular and D2D coverage probabilities are derived using the tool of stochastic geometry. In order to analyze the coverage probabilities, the transmit powers $p_k$ of the D2D transmitters are assumed to be i.i.d. with cdf $F_{p_k}(\cdot)$, $k=1,\cdots,K'$, and the transmit power $p_0$ of the uplink cellular user is independent having distribution $F_{p_0}(\cdot)$.

\subsection{Cellular Link Coverage Probability}
Based on the system model and assuming that the eNB is located at the origin, the SINR of the uplink is given by~\eqref{equ4}. We are interested in the cellular coverage probability $\covprobC{\beta_0}$, which is the probability that the SINR of cellular link  is greater than a minimum SINR $\beta_0$ for a reliable uplink connection as defined~\eqref{eq:cov_probC_def}. Using Lemma~\ref{lemma1}, we derive an analytical expression for the coverage probability of a cellular link.
\begin{proposition}\label{Cell_prob_cov} The cellular coverage probability is given by
\begin{equation}
\covprobC{\beta_0}=\mathbb{E}_X\left[ e^{-a_1 X-\theta_0 X^{2/\alpha}}\right],
\end{equation}
\end{proposition}
where $a_1=\beta_0\sigma^2$,
$\theta_0=\tfrac{\pi\mathbb{P}[Q_k=1]\lambda\beta_0^{2/\alpha}}{\text{sinc}(2/\alpha)}\mathbb{E}[p_k^{2/\alpha}]$, and $X=d_{0,c_m}^{\alpha} p_0^{-1}$ is a random variable with cdf $\newline {F_X(x)\!=\! \int\!\! F_{d_{0,c_m}}(x^{1/\alpha} p^{1/\alpha})dF_{p_0}(p)} $.
\begin{proof}
    See Appendix ~\ref{proof:Cell_prob}.
\end{proof}
One can note that the SINR of the uplink signal given in~\eqref{equ4} is independent of $d_{k,c_m }$; however, it depends on $K'$, which is the number of D2D users sharing the resource block with a particular uplink cellular user $c_m$. Therefore, the base station depends on how far the D2D users are from it and not how far the D2D users are from the cellular users; therefore, the joint probability distribution with respect to the random location of $c_m$'s is not needed when deriving the cdf of the SINR at the eNB.

The coverage probability depends on three D2D-related network parameters: $\mathbb{P}[Q_k=1]=q$, $\lambda$, and $\mathbb{E}[p_k^{2/\alpha} ] $. As the density $q\lambda$ of D2D transmitters decreases, $\covprobC{\beta_0}$  increases because a lower D2D link density causes less interference to the cellular link. Moreover, the random D2D PC parameter $p_k$, affects $\covprobC{\beta_0}$  only through its \nth{$({2}/{\alpha})$} moment.

Since the cellular user is uniformly distributed in a circle with center eNB and radius $R=\RC$, the cdf of the distance $d=d_{0,c_m}$ of the uplink is given by
\begin{equation}\label{equ:d00}
F_{d}(r)=\left\{
\begin{array}{c l}
0 & \text{if}\: r<0;\\
\tfrac{r^2}{R^2} & \text{if}\: 0\leq r \leq R;\\
1 & \text{if}\: r\geq R.
\end{array}
\right.
\end{equation}
Using  ~(\ref{equ:d00}), we consider the case when the uplink user employs a constant transmit power $p_0=P_{\max,\mathrm{C}}$, and assume a noise variance of $\sigma^2=0$ (so $\text{SINR}_0$ is reduced to $\text{SIR}_0$ (signal-to-interference ratio)). For a given path-loss exponent value, the coverage probability in the interference-limited regime becomes
\begin{align}\label{equ:Cov_prob_Cell}
\covprobC{\beta_0}&=\int_0^{\RC} \: \exp{\left( -\theta_0 r^{2} p_0^{-2/\alpha}\right)}\tfrac{2r}{\RC^2}dr
						  = \tfrac{1-\exp{\left(-\tfrac{ \mathbb{E}[K']\beta_0^{2/\alpha}}{\text{sinc}(2/\alpha)p_0^{2/\alpha}}\mathbb{E}\left[p_k^{2/\alpha} \right]\right)} }{\tfrac{\mathbb{E}[K'] \beta_0^{2/\alpha}}{\text{sinc}(2/\alpha)p_0^{2/\alpha}}\mathbb{E}\left[p_k^{2/\alpha} \right]},
\end{align}
where $\mathbb{E}[K']=\mathbb{P}[Q_k\!=\!1] \lambda \pi \RC^2$.

Expression~\eqref{equ:Cov_prob_Cell} explicitly shows that the coverage probability of the cellular link depends on: 1) the average number of active D2D transmitters $\mathbb{E}[K']$, 2) certain moments of the power transmitted from the D2D transmitters, 3) the power transmitted by the cellular user $p_0$, 4) path-loss exponent $\alpha$, and 5) the target SINR threshold $\beta_0$.


%
\subsection{D2D Link Coverage Probability}
Using the same approach in the previous subsection, the SINR of the $\nth{k}$ D2D link, based on the system model, is given in~\eqref{equ:SINR_k}. Then:

\begin{proposition}\label{D2D_prob_cov} The D2D coverage probability is given by
\begin{align}\label{eq:cov_probD_expec}
\covprobD{\beta_k}=\mathbb{E}_Z\left[ e^{-a_2Z-\theta_k Z^{2/\alpha}}\mathcal{L}_Y(\beta_k Z)\right],
\end{align}
where $\beta_k$ is the minimum SINR required for reliable transmission, $a_2=\beta_k\sigma^2$,
\newline$\theta_k= \tfrac{\pi\mathbb{P}[Q_k=1]\lambda\beta_k^{2/\alpha}}{\text{sinc}(2/\alpha)}\mathbb{E}[p_k^{2/\alpha}]$,
$Z=d_{k,k}^{\alpha} p_k^{-1}$ is a random variable with cdf \newline$F_Z(z)= \int F_{d_{k,k}}(x^{1/\alpha} p^{1/\alpha})dF_{p_k}(p)$, $Y =|h_{k,c_m}|^2d_{k,c_m}^{-\alpha} p_0$, and ${\mathcal{L}_Y(\beta_k Z)= \mathbb{E}_Y[ e^{-(\beta_k Z)Y}]}$.
\end{proposition}
\begin{proof}
See Appendix~\ref{proof:D2D_prob}.
\end{proof}
Using the fact that $|h_{k,c_m}|^2\sim \exp{(1)}$, which implies $\mathbb{P}(|h_{k,c_m}|^2\geq x)= e^{-x}$, and the expectation is over $d_{k,c_m}$ in $\mathcal{L}_Y(\beta_k Z)$, we derive a closed form expression for the  D2D coverage probability~\eqref{eq:cov_probD_expec} in an interference-limited regime (where noise variance $\sigma^2=0$, and $\text{SINR}_k$ reduces to $\text{SIR}_k$) as
\begin{align}
    \covprobD{\beta_k} &=
        \mathbb{E}_{d_{k,k}^{\alpha} p_k^{-1}} \left[ \exp{\left(- \theta_k \left(d_{k,k}^{\alpha}p_k^{-1}\right)^{2/\alpha} \right)} \times \right. \left. \mathbb{E}_{d_{k,c_m}}\left[ e^{-\beta_k \left( d_{k,k}^{\alpha}p_k^{-1}\right)|h_{k,c_m}|^2d_{k,c_m}^{-\alpha} p_0}\right] \right] \notag\\
                        &=
        \mathbb{E}_{d_{k,k}^{\alpha} p_k^{-1}} \left[ \exp{\left(- \theta_k \left( d_{k,k}^{\alpha}p_k^{-1} \right)^{2/\alpha} \right)} \times \right.  \left.\mathbb{E}_{d_{k,c_m}}\left[\frac{1}{1+\beta_k\frac{p_0}{p_k }d_{k,k}^{\alpha}d_{k,c_m}^{-\alpha}} \right] \right]. \label{equ:Cov_prob_D2D_ana_origin}
\end{align}
We next simplify~\eqref{equ:Cov_prob_D2D_ana_origin} by deriving expressions for the various expectations involved.
\begin{corollary}\label{cor:expec_dk} Using Lemma~\ref{lemma1} and considering the proposed channel allocation scheme for the case of 2 CUEs, then the first moment of the distance between two uniformly distributed points can be approximated as $\mathbb{E} \left[d_{k,c_m}\right]\approx 512\RC/(45\pi^2)$.
\end{corollary}
\begin{proof}
See Appendix ~\ref{proof:d_k,cm}.
\end{proof}
We next employ the following approximation $ \mathbb{E}_{d_{k,c_m}}\bigl[\tfrac{1}{1+\kappa d_{k,c_m}^{-\alpha}}\bigr] \approx \tfrac{1}{1+\left(\kappa\right)^{2/\alpha} \mathbb{E} \left[d_{k,c_m}\right]^{-2}} \vspace{2pt}$ as in~\cite{lee2015power}. Using this approximation together with the result from corollary~\ref{cor:expec_dk}, equation~\eqref{equ:Cov_prob_D2D_ana_origin} reduces to
\begin{equation}\label{equ:Cov_prob_D2D_ana_app_origin}
\begin{split}
\covprobD{\beta_k}
&\approx \mathbb{E}_{d_{k,k}^{\alpha} p_k^{-1}} \left[ \tfrac{\exp{\left(- \theta_k \left( d_{k,k}^{\alpha}p_k^{-1}\right)^{2/\alpha} \right)}}{1+\left( \beta_k\tfrac{p_0}{ p_k d_{k,k}^{-\alpha}}\right)^{2/\alpha} (512\RC/(45\pi^2))^{-2}} \right].
\end{split}
\end{equation}

\subsection{Discussion}
The coverage probability depends on the following D2D-related network parameters: density of the D2D links ($\lambda$), thinning probability $q$, target SINR ($\beta$), the moments of the power transmitted from the D2D transmitters, and the power transmitted by the cellular user. This modeling approach allows us to analyze the coverage probability and ergodic rate for a D2D underlaid cellular network with high accuracy. It also enables network designers/operators to optimize network performance by efficiently determining the optimal network parameters mentioned above. The system can control the impact of D2D links on the cellular link through 1) the proposed channel allocation scheme, which constrains the density of the D2D links that uses the same resources with a particular cellular user, and 2) through the proposed power control schemes, which control the transmit power of the D2D users.

%
\section{Proposed D2D Distributed Power Control Schemes}\label{s:pc_schemes}
 In order to minimize the interference caused by the D2D users, we propose distributed power control schemes that only require the CSI of the direct link. For link establishment, two static distributed PC are proposed, and both rely on the distance-dependent path-loss parameters~\cite{Globecomm2017_mansour_b,Wcnc2018_mansour_a}. On the other hand, for link maintenance, a more adaptive distributed PC is proposed that compensates the measured SINR at the receiver with a variable target SINR.

%
\subsection{Proposed Distance-based Path-loss Power Control (DPPC)}
In this PC scheme, each D2D transmitter selects its transmit power based on the channel conditions, namely the distance-based path-loss $d_{k,k}^{-\alpha}$, so as to maximize its own D2D link rate. In order to realize our proposed scheme, we define D2D proximity as the area of a disk centered at the transmitting UE, with radius $R_{\mathrm{D}}={\left( \tfrac{P_{\max,\mathrm{D}}}{\rho_\mathrm{rx}}\right) }^{1/\alpha}$, where $P_{\max,\mathrm{D}}$ is the maximum transmit power of the D2D UE, and $\rho_\mathrm{rx}$ is the minimum power for the D2D RX to recover a signal (sometimes referred to as receiver sensitivity).

The $\nth{k}$ D2D TX can only use transmit power $p_k$ with transmit probability $\txprob$, if the channel quality of the $\nth{k}$ D2D link is favorable, in the sense that it exceeds a known non-negative threshold $\Gamma_{\min}$:
\begin{align}\label{Fav:Ptx}
\txprob &\triangleq \prob{\abs{h_{k,k}}^2 d_{k,k}^{-\alpha} \geq \Gamma_{\min}} \approx  \exp{\left(-\Gamma_{\min}\: \expec{d_{k,k}^{\alpha}}\right)}.
\end{align}
Furthermore, an estimation error margin $\varepsilon$ is introduced to compensate for the error in estimating the distance between the D2D pairs. Hence, the proposed power allocation, based on the \textit{channel inversion} for the D2D link, is given by
\begin{align}\label{equ:DPPC}
p_{k}=\left\{
        \begin{array}{c l}
        \rho_\mathrm{rx} d_{k,k}^{\alpha}(1+\varepsilon) & \text{with}\: \txprob,\\
        0 & \text{with}\: 1-\txprob,
        \end{array}
        \right.
\end{align}
where $d_{k,k}$ is the distance between the $\nth{k}$ D2D pair, $\alpha$ is the path-loss exponent, and $\varepsilon$ is the estimation error margin of $d_{k,k}^{\alpha}$, such that $0 \leq \varepsilon \leq 1$.

Each D2D transmitter decides its transmit power based on its own channel gain and a known non-negative threshold $\Gamma_{\text{min}}$. For a given distribution of the channel gain, selecting a proper threshold $\Gamma_{\text{min}}$ plays an important role in determining the sum rate performance of the D2D links. For instance, if a large $\Gamma_{\text{min}}$ is chosen (implying a small $\txprob$), the inter-D2D interference is reduced. However, a larger $\Gamma_{\text{min}}$ (implying a smaller $\txprob$) means a smaller number of active D2D links within the cell. Thus, $\Gamma_{\text{min}}$ needs to be carefully chosen to balance these two conflicting factors, while providing a high D2D sum rate. We optimize the choice $\Gamma_{\text{min}}$ so as to maximize the D2D sum rate as discussed in Section~\ref{ss:DPPC_thresold}.

Moreover, the D2D transmitter checks if the link quality degrades (i.e., $\abs{h_{k,k}}^2 d_{k,k}^{-\alpha}\!<\!\Gamma_{\text{min}}$), then the D2D communication is dropped. Also, the D2D receiver checks if the estimated distance-based path-loss increases, and reports it to the D2D transmitter, conditioned on the fact that the D2D communication link remains active if this distance remains within $R_{\mathrm{D}}$.

Note here that channel inversion only compensates for the large-scale path-loss effects and not for small-scale fading effects. For instance, instantaneous CSI is not required at the transmitter, since the loss due to distance is compensated. Moreover, the proposed scheme captures the randomness of the distance between the D2D pairs, and if the D2D pairs are close to each other, they will allocate less power than the case if they are further apart. However in \cite{lee2015power}, a fixed distance between the pairs is assumed and maximum power $P_{\max,\mathrm{D}}$ is always allocated for D2D transmission, which needlessly increases power consumption and generates more interference.

Considering the proposed DPPC scheme along with the random locations of D2D users, the transmit powers and the SINRs experienced by the receivers become random as well. Therefore in what follows, we first characterize the transmit power $p_k$ via its $\nth{\alpha/2}$ moment, and then characterize the cellular and D2D coverage probabilities accordingly. Finally, we derive an expression for the D2D sum rate and maximize in order to optimize the DPPC threshold $\Gamma_{\text{min}}$.

%
\subsubsection{Analysis of Power Moments}\label{s:power_moments}
According to the system model, the D2D receivers are considered to be uniformly distributed in a circle centered at the corresponding D2D transmitter with radius $R_{\mathrm{D}}$; therefore, the cdf of the distance $d_{k,k}$ of the D2D link is similar to that of $d_{0,c_m}$ in~\eqref{equ:d00}, where $d=d_{k,k}$ and $R=\RD$. Using~\eqref{equ:d00}, the moments of the transmit power $p_k$ for the DPPC scheme, where $p_k=\rho_\mathrm{rx} {d_{k,k}}^{\alpha}(1+\varepsilon)$, can be expressed as
\begin{align}\label{equ:Epk}
    \mathbb{E}_{d_{k,k}}\left[p_k^{2/\alpha} \right]
        &\!=\! \rho_\mathrm{rx}^{2/\alpha}\! \int_0^{R_{\mathrm{D}}} \!r^2\tfrac{2r}{R_{\mathrm{D}}^2} (1\!+\!\varepsilon)^{2/\alpha} dr
        \!=\!\rho_\mathrm{rx}^{2/\alpha}\: \tfrac{R_{\mathrm{D}}^2}{2}(1+\varepsilon)^{2/\alpha}.
\end{align}

\textit{Cellular Coverage Probability for DPPC:} By substituting~\eqref{equ:Epk} for $\mathbb{E}_{d_{k,k}}\left[p_k^{2/\alpha}\right]$ into the derived expression~\eqref{equ:Cov_prob_Cell}, the cellular coverage probability for DPPC can be obtained.

\textit{D2D Coverage Probability for DPPC:} For $p_k=\rho_\mathrm{rx} {d_{k,k}}^{\alpha}(1+\varepsilon)$, and using the moments of $p_k$ in~\eqref{equ:Epk}, the D2D coverage probability in~\eqref{equ:Cov_prob_D2D_ana_origin} becomes
\begin{align}\label{equ:Cov_prob_D2D_ana}
\covprobD{\beta_k}&=e^{- \theta_k  \left( \rho_\mathrm{rx}(1+\varepsilon)\right)^{-2/\alpha}} \times \mathbb{E}_{d_{k,c_m}}\left[\dfrac{1}{1+\beta_k\tfrac{p_0}{
                          \rho_\mathrm{rx}(1+\varepsilon)}d_{k,c_m}^{-\alpha}}\right],
\end{align}
Following the same approach as in~\eqref{equ:Cov_prob_D2D_ana_app_origin}, the approximated expression for $\covprobD{\beta_k}$ is given by
\begin{align}\label{equ:Cov_prob_D2D_ana_app}
\covprobD{\beta_k}\approx \dfrac{e^{- \theta_k \left( \rho_\mathrm{rx}(1+\varepsilon)\right)^{-2/\alpha}}}{1+\left( \beta_k\tfrac{p_0}{ \rho_\mathrm{rx}(1+\varepsilon)}\right)^{2/\alpha} \left(512\RC/(45\pi^2)\right)^{-2}}.
\end{align}

%
\subsubsection{Sum Rate of D2D Links}
We analyze the sum rate of D2D links when the proposed DPPC scheme is employed, and compute the optimal threshold $\Gamma_{\text{min}}$ of the proposed PC that maximizes the sum rate of D2D links.

Let $|A_\mathrm{D}|$ denote the number of active links selected by the proposed PC and CA algorithms, i.e., $|A_\mathrm{D}|= \mathbb{P}[Q_k=1] \times \mathbb{P}[|h_{k,k} |^2d_{k,k}^{-\alpha} \geq \Gamma_{\text{min}}]\lambda\pi \RC^2 = \tilde{\lambda}\pi \RC^2 $, where $\tilde{\lambda}=\mathbb{P}[Q_k=1] \times \mathbb{P}[|h_{k,k} |^2d_{k,k}^{-\alpha} \geq \Gamma_{\text{min}}]\lambda=q \txprob\lambda$. As in \cite{lee2015power}, we assume Gaussian signal transmissions on all links, and hence, the distribution of the interference terms becomes Gaussian.

From the \text{SIR} distribution of the D2D link given in (\ref{equ:Cov_prob_D2D_ana_app}) with $\sigma^2=0$, the ergodic rate of the typical D2D link is generally expressed as
\begin{align}\label{equ:erg_rate_D2D_analytical}
\bar{R}_{D2D}   &=\int_0^\infty \text{log}_2 \left (1 + x \right)\frac{\partial}{\partial x}\left[ \mathbb{P}[\text{SIR}_k \geq x] \right] dx \approx \int_0^\infty \frac{1}{1+x}\covprobD{x} \text{d}x \notag \\
                &\approx \int_0^\infty \dfrac{\exp{\left(- \tfrac{\pi\tilde{\lambda} x^{2/\alpha}}{\text{sinc}(2/\alpha)}\mathbb{E}\left[p_k^{2/\alpha} \right] \left( d_{k,k}^{\alpha}p_k^{-1} \right)^{2/\alpha}  \right)}}{({1+x} )\times \left( 1+\left( x \tfrac{p_0}{d_{k,k}^{-\alpha}p_k }\right)^{2/\alpha} \mathbb{E} \left[d_{k,c_m}\right]^{-2} \right)}\text{d}x.
\end{align}
Note that the above general expression of the ergodic rate is valid for any distributed power control scheme that allocates its own transmit power independently of the transmit power used at other D2D transmitters.

Using~\eqref{equ:rate} and~\eqref{equ:erg_rate_D2D_analytical}, the new achievable sum rate of D2D links is given as
\begin{align}\label{equ:sum_rate_D2D_analytical}
\mathcal{R}_s^{(D)} &=\mathbb{E} \left[\sum_{k=1}^{K'} \text{log}_2\left (1 + \text{SIR}_k \right)\right] =|A_\mathrm{D}| \times \bar{R}_{D2D}=\tilde{\lambda}\pi \RC^2  \times \bar{R}_{D2D}.
\end{align}

%
\subsubsection{D2D Power Control Threshold for DPPC}\label{ss:DPPC_thresold}
From the ergodic sum rate of D2D links, we optimize the D2D PC threshold $\Gamma_{\text{min}}$ by maximizing the derived transmission capacity of D2D links, which is given as
\begin{equation}\label{equ:sum_rate_D2D_analytical2}
\begin{split}
\mathcal{R}_s^{(D)} (\beta_k) &\approx  \tfrac{q\lambda \txprob \pi \RC^2 \: \text{log}_2(1+\beta)}{1+\left( \beta_k\dfrac{p_0}{\left( \rho_\mathrm{rx}(1+\varepsilon)\right)}\right)^{2/\alpha} \left( \tfrac{512\RC}{45\pi^2}\right) ^{-2}} \times \exp{\left(- \tfrac{\pi \tilde{\lambda}\beta_k^{2/\alpha}}{\text{sinc}(2/\alpha)}\mathbb{E}\left[p_k^{2/\alpha} \right] \left( \rho_\mathrm{rx}(1+\varepsilon)\right)^{-2/\alpha} \right)}\\
        &\approx \tfrac{\tilde{\lambda}\pi \RC^2 \: \text{log}_2(1+\beta_k)}{1+\kappa \beta_k^{2/\alpha}} \exp{\left(- \dfrac{\pi \tilde{\lambda}\beta_k^{2/\alpha}}{\text{sinc}(2/\alpha)} \tfrac {\RD^2}{2} \right)},
\end{split}
\end{equation}
where $\kappa=  \left(\tfrac{p_0}{\left( \rho_\mathrm{rx}(1+\varepsilon)\right)}\right)^{2/\alpha} \left( \tfrac{512R}{45\pi^2}\right) ^{-2}$ and $\tilde{\lambda}=q\lambda \txprob$. By solving the following optimization problem, we can compute the new optimal transmission probability:
\begin{align*}
{\text{maximize}}\quad & \mathcal{R}_s^{(D)} (\beta) \\
\text{subject to}\quad & 0\geq\txprob\geq 1
\end{align*}
The optimal solution of $\txprob$ can be obtained by the $1^{st}$ order optimality solution, since the objective function has one optimum point. The first order derivative yields:
\begin{equation}\label{equ:deriv_Ps}
 \tfrac{d \mathcal{R}_s^{(D)} (\beta_k)}{d \txprob}= 1- \tfrac{\pi q\lambda \beta_k^{2/\alpha} \tfrac {\RD^2}{2}}{\text{sinc}(2/\alpha)} \txprob=0.
\end{equation}
 The second derivative of $\mathcal{R}_s^{(D)} (\beta_k)$ is applied to test the concavity at $\txprob$, which is given as
\begin{equation}\label{equ:deriv_Ps2}
 \tfrac{d^2 \mathcal{R}_s^{(D)} (\beta_k)}{{d \txprob}^2}= - \tfrac{\pi q\lambda \beta_k^{2/\alpha} \tfrac {\RD^2}{2}}{\text{sinc}(2/\alpha)} <0 \:\:\: \text{    for } \alpha\geq 2.
\end{equation}
Thus, $\mathcal{R}_s^{(D)} (\beta_k)$ is maximum at $\txprob=\tfrac{2\text{sinc}(2/\alpha)}{\pi q\lambda \beta_k^{2/\alpha} {\RD^2}}$. However, to satisfy the conditions of $\txprob \in \{0,1\}$, we have $\displaystyle \txprob^\star= \min{\Big\{ \tfrac{2\text{sinc}(2/\alpha)}{\pi q\lambda \beta_k^{2/\alpha} {\RD^2}},1\Big\}}\vspace{2pt}$. Using ~(\ref{equ:DPPC}) where $\txprob=\exp{ (-\Gamma_{\text{min}}\: \mathbb{E}[ d_{k,k}^\alpha] )}$, then the optimal threshold $\Gamma_{\text{min}}^{\star}$ can be obtained as
\begin{equation}\label{equ:Gmin}
\Gamma_{\text{min}}^{\star}= - \ln{(\txprob^{\star})} \frac{2+\alpha}{2} R_{\mathrm{D}}^{-\alpha}
\end{equation}

Knowing the solution $\txprob^{\star}$, the approximated transmission capacity in~\eqref{equ:sum_rate_D2D_analytical2} can be rewritten as
\begin{equation}\label{equ:trans-cap}
\mathcal{R}_s^{(D)} (\beta) \approx \left\{
\begin{array}{l l}
\dfrac{\lambda\pi \RC^2 \: \text{log}_2(1+\beta_k)}{1+\kappa \beta_k^{2/\alpha}}  {\exp{\left(- \dfrac{\pi q\lambda \beta_k^{2/\alpha} R_{\mathrm{D}}^2}{2 \text{sinc}(2/\alpha)} \right)}} & \text{for} \:\beta_k \leq \tilde{\beta_k},\\
\dfrac{{2 \text{sinc}(2/\alpha)}\RC^2 \: \text{log}_2(1+\beta_k)}{\beta_k^{2/\alpha}R_{\mathrm{D}}^2  (1+\kappa \beta^{2/\alpha})\exp{(1)} } & \text{for} \: \beta_k > \tilde{\beta_k},\\
\end{array}
\right.
\end{equation}
where $\tilde{\beta_k}= \left[\tfrac{2 \text{sinc}(2/\alpha)}{\pi q\lambda R_{\mathrm{D}}^2 } \right]^{\alpha/2} $.

The transmission capacity of the D2D links depends on the relationship between the minimum SINR value $\beta_k$ and the network parameters: path-loss exponent $\alpha$, the density of the D2D links $q\lambda$, and the maximum allowable distance between the D2D pairs $R_{\mathrm{D}}$. When $\beta_k<\tilde{\beta_k}$, all D2D transmitters are scheduled; therefore no admission control is applied. However, when $\beta_k \geq \tilde{\beta_k}$, the D2D links are scheduled with transmit probability $\txprob^{\star}$, which mitigates the inter-D2D interference since the transmission capacity no longer depends on the density of the nodes $\lambda$.

By integrating the transmission capacity in ~(\ref{equ:trans-cap}) with respect to $\beta$, the sum rate of D2D links is expressed as follows
\begin{equation}\label{equ:sum_rate_D2D_analytical_final}
\begin{split}
\mathcal{R}_s^{(D)} &\approx \int_0^{\tilde{\beta_k}} {\dfrac{q\lambda\pi \RC^2 }{(1+x)(1+\kappa x^{2/\alpha})}  {\exp{\left(- \dfrac{\pi q\lambda x^{2/\alpha} R_{\mathrm{D}}^2}{2 \text{sinc}(2/\alpha)} \right)}}}\text{d}x +\int_{\tilde{\beta_k}}^{\infty} {\dfrac{{2 \text{sinc}(2/\alpha)}\RC^2}{(x^{2/\alpha}R_{\mathrm{D}}^2 ) (1+x) (1+\kappa x^{2/\alpha})} \exp{(-1)}}\text{d}x.\\
\end{split}
\end{equation}
The DPPC scheme is summarized in the first part of the pseudo-code in Algorithm~\ref{alg:DPPC}.

\begin{algorithm}
\caption{Static Distributed Power Control}\label{alg:DPPC}\small
\begin{algorithmic}[1]
\If {D2D TX $k$ is unable to acquire $d_{0,k}$ }
\LineComment{\emph{Apply} \textbf{\textit{DPPC}} \emph{scheme}}
\State Calculate $\Gamma_{\text{min}}$ that maximizes the D2D sum rate $\mathcal{R}_s^{(D)}(\beta)$ according to ~(\ref{equ:Gmin})
\If {$|h_{k,k} |^2d_{k,k}^{-\alpha} \geq \Gamma_{\text{min}}$ \textbf{and} $d_{k,k} \leq R_{\mathrm{D}}$}
\State D2D candidates transmit in D2D mode
\State $p_k \gets \rho_\mathrm{rx} d_{k,k}^{\alpha}(1+\varepsilon)$ .
\Else {\: $p_k \gets 0$}
\EndIf
\Else
\LineComment{\emph{Apply} \textbf{\textit{EDPPC}} \emph{scheme}}
\State Set $\Gamma_{\text{min}}=G_{\min}$
\If {$|h_{k,k} |^2d_{k,k}^{-\alpha} \geq \Gamma_{\text{min}}$ \textbf{and} $d_{k,k} \leq R_{\mathrm{D}}$}
\State D2D candidates transmit in D2D mode
\State $U \gets \rho_\mathrm{rx}(1+\varepsilon)$, $V \gets \mu\rho_\mathrm{rx}(1+\varepsilon)$
\State $p_k \gets \min\{ U d_{k,k}^{\alpha}, V d_{0,k}^{\alpha}\}$ .
\Else {\: $p_k\gets0$}
\EndIf
\EndIf
\end{algorithmic}
\end{algorithm}
\normalsize

%
\vspace{-0.2in}
\subsection{Proposed Extended Distance-based Path-loss Power Control (EDPPC)}
EDPPC is proposed as an extended DPPC scheme for link establishment stage. We consider in this scheme an extra distance-based path-loss parameter $d_{0,k}^{-\alpha}$, where $d_{0,k}$ is the distance between the eNB and the D2D $\nth{k}$ TX, in order to reduce the D2D interference at the eNB. This scheme works only if the D2D users are able to obtain estimates of $d_{0,k}$ from the eNB.

We apply the same conditions as in DPPC in~\eqref{Fav:Ptx}, where the $\nth{k}$ D2D TX can only use the transmit power $p_k$ with transmit probability $\txprob$ for favorable channel conditions. However, in this PC scheme, $\Gamma_\text{min}=G_{\text{min}}$ is a static value that is chosen by the eNB and broadcasted to the D2D transmitters.

The EDPPC scheme works as follows: each D2D TX selects its transmit power based on the distance-based path-loss parameters $d_{k,k}^{-\alpha}$ and $d_{0,k}^{-\alpha}$. The role of the additional parameter $d_{0,k}^{-\alpha}$ is to suppress interference even more at the eNB. Let $U = \rho_\mathrm{rx}(1+\varepsilon)$ and $V = \mu\rho_\mathrm{rx}(1+\varepsilon)$, where $\mu$ is a PC parameter with small value chosen so that the D2D transmitter does not cause excessive interference to the eNB and to other D2D UEs in the same cell, and $\varepsilon$ is an estimation error margin that offsets any inaccuracy in estimating the path-loss parameters $d_{k,k}^{\alpha}$ and $d_{0,k}^{\alpha}$. Then, the proposed power allocation for the D2D link is based on the following:
\begin{equation}
	p_{k} = \left\{
			\begin{array}{c l}
				\min\{ U d_{k,k}^{\alpha}, V d_{0,k}^{\alpha}\} & \text{with}\: \txprob\\
				0 & \text{with}\: 1-\txprob,
			\end{array}
			\right.\label{equ:MDPPC}
\end{equation}

Due to the nature of the EDPPC scheme, along with the random locations of D2D users, the transmit powers and the SINRs experienced by the receivers become also random. Therefore, we derive $\nth{\alpha/2}$ moments of the transmit power $p_k$ so that the cellular and D2D coverage probabilities can be characterized accordingly.

%
\subsubsection{Analysis of Power Moments}
The D2D TX and the corresponding D2D RX are assumed to be uniformly distributed; therefore, the distance $d_{0,k}$ of the D2D interfering link with the eNB and the distance $d_{k,k}$ of the direct D2D link are uniformly distributed in circles with radii $\RC$ and $\RD$, respectively.
\begin{theorem}\label{Theorem_EDPPC}
The expected value of the minimum of two random variables $A,B \in \Omega \rightarrow \mathsf{R}$ is $\expec{\min(A,B)}=\mathbb{E}[A]+\mathbb{E}[B]- \mathbb{E}[\max(A,B)]$.
\end{theorem}
\begin{proof}
See Appendix ~\ref{proof:Therem:EDPPC}.
\end{proof}
\begin{lemma}\label{Lemma_EDPPC}
The expected value of the maximum of two random variables $A,B \in \Omega \rightarrow \mathsf{R}$ with pdfs $f_A(a),f_B(b)$ and cdfs $F_A(a),F_B(b)$, respectively, is
\begin{equation}\label{eq:expec_maxAB}
\mathbb{E}[\max(A,B)]=\int_{-\infty}^{\infty}a f_{A}(a)F_B(a)da +\int_{-\infty}^{\infty}b f_{B}(b)F_A(b)db.
\end{equation}
\end{lemma}
\begin{proof}
See Appendix~\ref{proof:LemmaEDPPC}.
\end{proof}
\begin{corollary}\label{Corollary:EDPPC}
Using the distribution functions of $d_{k,k}$ and $d_{0,k}$, the moments of the transmit power $p_k$ are given by
\begin{equation}\label{equ:Epk2}
\mathbb{E}_{d_{k,k}}\left[p_k^{2/\alpha} \right]=\left\{
\begin{array}{c l}
 \frac{R_{\mathrm{C}}^2 V^{2/\alpha}}{2}- \frac{ R_{\mathrm{C}}^4 V^{4/\alpha}}{6R_{\mathrm{D}}^2 U^{2/\alpha}} & \text{if}\: R_{\mathrm{D}}^2 U^{2/\alpha} > R_{\mathrm{C}}^2 V^{2/\alpha}\\
 \frac{R_{\mathrm{D}}^2 U^{2/\alpha}}{2}- \frac{ R_{\mathrm{D}}^4 U^{4/\alpha}}{6R_{\mathrm{C}}^2 V^{2/\alpha}} & \text{if}\: R_{\mathrm{D}}^2 U^{2/\alpha} \leq R_{\mathrm{C}}^2 V^{2/\alpha}.
\end{array}
\right.
\end{equation}
\end{corollary}
\begin{proof}
See Appendix~\ref{Proof:COr_EDPPC}.
\end{proof}

Under this power control scheme, it is noted that: 1) D2D UEs closer to the serving eNB (where $d_{0,k} < d_{k,k}$) normally cause a stronger uplink interference and thus their transmit powers are reduced, 2) D2D UEs closer to the cell edge can transmit at a higher power since their interference to the uplink cellular UE is dropped due to path-loss, and 3) D2D pairs with close proximity will be allocated less power than D2D pairs that are far apart.

The EDPPC scheme is summarized in the second part of Algorithm~\ref{alg:DPPC}.

\textit{Cellular Coverage Probability for EDPPC:} By substituting $\mathbb{E}_{d_{k,k}}\left[p_k^{2/\alpha}\right]$ obtained in ~(\ref{equ:Epk2}) into the derived expressions ~(\ref{equ:Cov_prob_Cell}), the cellular coverage probability for EDPPC can be obtained.

\textit{D2D Coverage Probability for EDPPC:} Using the same methodolgy as in Theorem~\ref{Theorem_EDPPC}, for $p_k=\min\{ U d_{k,k}^{\alpha}, V d_{0,k}^{\alpha}\}$, and using the moments of $p_k$ in~\eqref{equ:Epk2} and the pdf of $d_{k,k}$ and $d_{0,k}$, the D2D coverage probability in~\eqref{equ:Cov_prob_D2D_ana_app_origin} becomes
\begin{equation}
\begin{split}			
			\covprobD{\beta_k}\approx &\tfrac{e^{- \theta_k \left( \rho_\mathrm{rx}(1+\varepsilon)\right)^{-2/\alpha}}}{1+\left( \beta_k\tfrac{p_0}{ \rho_\mathrm{rx}(1+\varepsilon)}\right)^{2/\alpha} \left(512\RC/(45\pi^2)\right)^{-2}} \int_{0}^{\mu^{1/\alpha}\RC}\! \left(  \int_{0}^{x} \tfrac{2y}{\RD^2}dy \right) \tfrac{2\mu^{-2/\alpha}x}{\RC^2}dx \:\:+ \\
			&\int_{0}^{\RD}\! \left(\int_{0}^{y} \tfrac{\exp{\left(- \theta_k \left( x^{-\alpha}y^{\alpha}\left( \rho_\mathrm{rx}(1+\varepsilon)\right)^{-1}\right)^{2/\alpha} \right)}}{1+\left( \beta_k\tfrac{p_0}{ x^{\alpha}y^{-\alpha}\left( \rho_\mathrm{rx}(1+\varepsilon)\right)}\right)^{2/\alpha} (512\RC/(45\pi^2))^{-2}} \tfrac{2\mu^{-2/\alpha} x}{\RC^2}dx\right)\tfrac{2 y}{\RD^2}dy.
\end{split}\label{equ:MDPPC_PDCov}
\end{equation}

To validate our analysis for DPPC and EDPPC, we compare the derived analytical expressions with their corresponding simulated results for $\lambda \in \{2\times 10^{-5},5\times 10^{-5}\}$, $\tilde{\lambda}=0.5\lambda$, $M=2$, $\mu=0.0005$, $\varepsilon=0.5$, and $\alpha=4$. In Fig.~\ref{fig:DPPC_ANA_SIM} and Fig.~\ref{fig:MIN_ANA_SIM}, we validate the correctness of the analytical expressions for the cellular coverage probability of~\eqref{equ:Cov_prob_Cell} and D2D coverage probability of~\eqref{equ:Cov_prob_D2D_ana_app} and ~\eqref{equ:MDPPC_PDCov}, while using the derived expressions of $\mathbb{E}[p_k^{2/\alpha}]$ for DPPC and EDPPC in ~(\ref{equ:Epk}) and ~(\ref{equ:Epk2}), respectively. As shown in the plots, the curves of the proposed DPPC and EDPPC schemes match well with simulated results over the entire range of $\beta$.
\begin{figure*}[hbtp]
\centering
\subfigure[]{
    \includegraphics[width=0.5\textwidth]{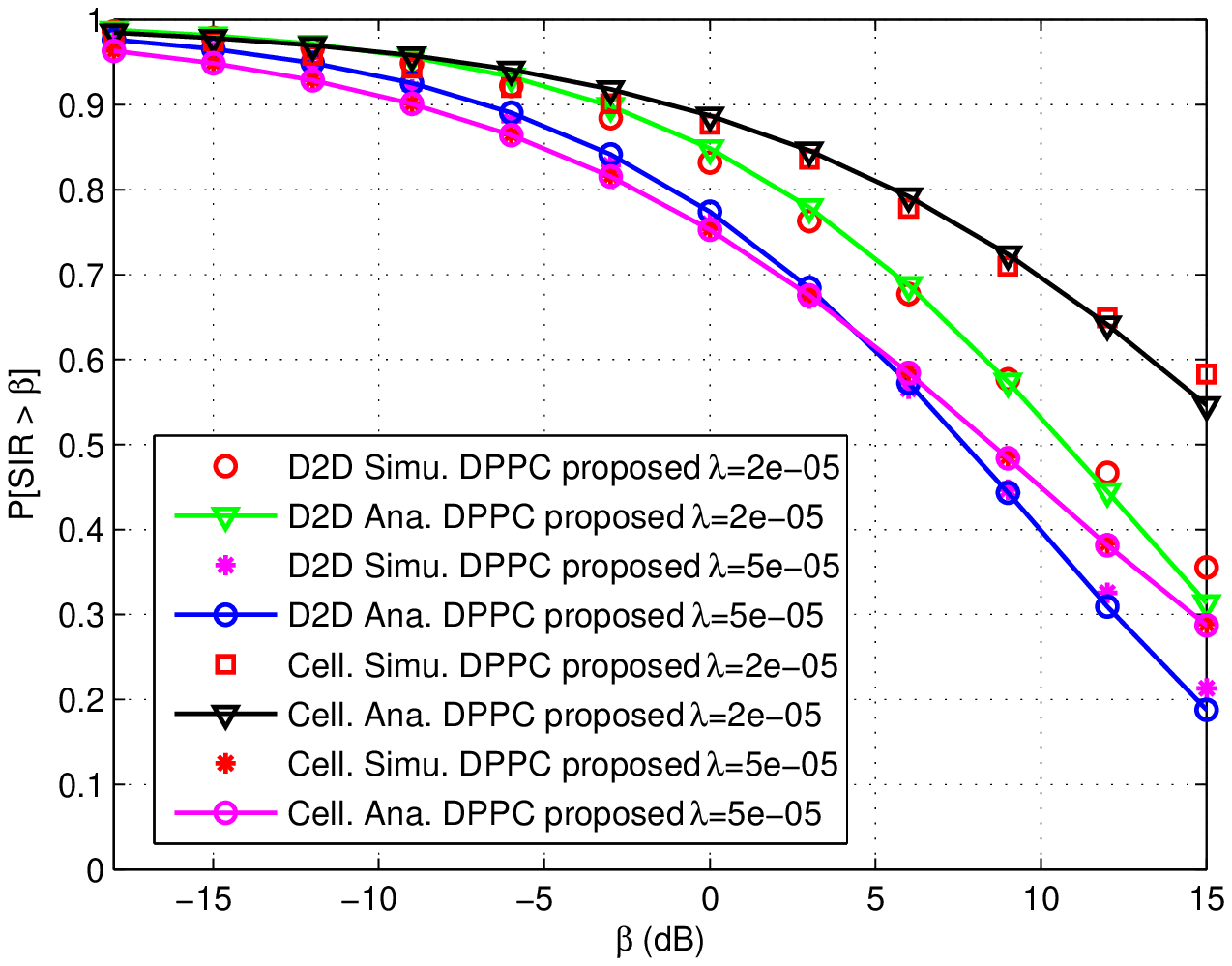}
	\label{fig:DPPC_ANA_SIM}}\vspace{-0.125in}\hspace{-0.5in}
\quad
	\subfigure[]{
	\includegraphics[width=0.5\textwidth]{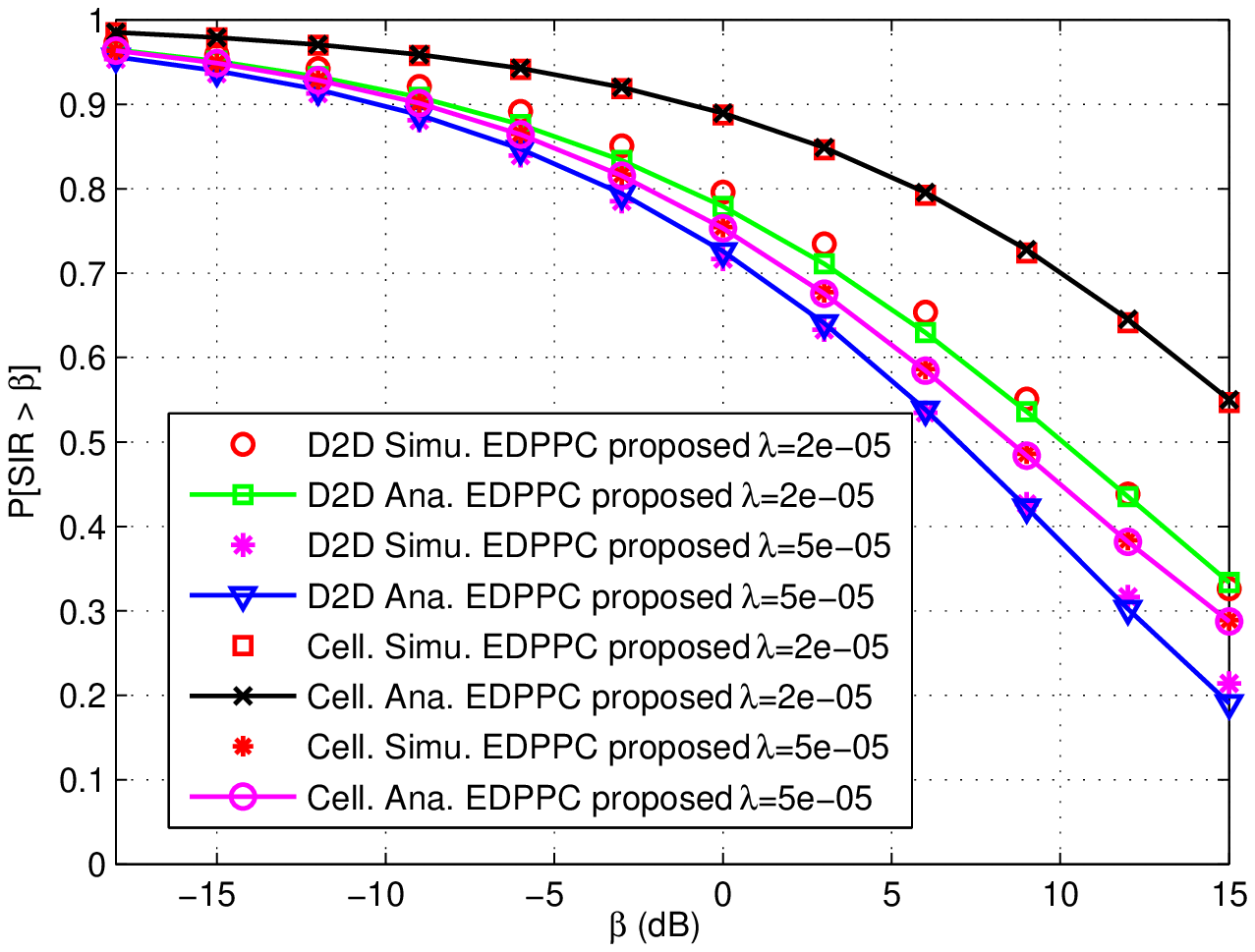}
\label{fig:MIN_ANA_SIM}}
\caption{Analytical vs. simulated coverage probability for cellular and D2D users using (a) DPPC, and (b) EDPPC scheme.}
\label{fig:SImfig}
\end{figure*}

%
\subsection{Proposed Soft Dropping Distance-based Power Control (SDDPC)}
The PC schemes proposed earlier provide a static power allocation where varying channel quality during D2D transmissions is not taken into consideration. An adaptive PC with variable target SINR would be an attractive approach to guard cellular and D2D communications against mutual interference and maintain good link quality. We propose a soft dropping distance-based PC (SDDPC) scheme that gradually decreases the target SINR as the required transmit power increases. This increases the probability of finding a feasible solution for the PC problem in which the target SINR values for all co-channel links can be achieved. Hence, links with bad quality, where the receiver is far from the transmitter and requires higher power, would target \emph{lower} SINR values. On the other hand, links with better quality, where the receiver is near the transmitter and requires lower power, would target \emph{higher} SINR values.

In the SDDPC scheme, each UE iteratively varies its transmit power so that a power vector $\textbf{p}$ for all UEs in the system is found such that the $\text{SINR}_{k}$ of the $\nth{k}$ UE satisfies
\begin{equation*}
 \text{SINR}_{k} (K',\textbf{p}) \geq\beta_{k}(d_{k,k}),
\end{equation*}
where $\beta_{k}(d_{k,k})$ is the target SINR of the $\nth{k}$ UE that varies according to the distance between the D2D pairs $d_{k,k}$. The SDDPC scheme uses a target SINR that varies between a maximum value $\beta_{\text{max}}$ and a minimum $\beta_{\text{min}}$ as the distance between the D2D pairs varies between $R_{\min,\mathrm{D}}$ and a maximum value $R_{\mathrm{D}}$, while satisfying a power constraint of $P_{\min,\mathrm{D}} \leq p_k \leq P_{\max,\mathrm{D}}$.

The target SINR $\beta_{k}(d_{k,k})$  of the $\nth{k}$ D2D UE at TTI $(t)$ is given according to
\begin{equation}\label{Gamma}
\beta_{k}(d_{k,k})=\left\{
\begin{array}{c l}
\beta_{\text{max}} & \text{if}\: d_{k,k}^{(t)} \leq R_{\min,\mathrm{D}}\\
\beta_{\text{max}} \left( \dfrac{d_{k,k}^{(t)}}{R_{\min,\mathrm{D}}} \right)^\upsilon & \text{if}\:  R_{\min,\mathrm{D}}< d_{k,k}^{(t)}< R_{\mathrm{D}}\\
\beta_{\text{min}} & \text{if}\: d_{k,k}^{(t)} \geq R_{\mathrm{D}},
\end{array}
\right.
\end{equation}
where $\upsilon=\tfrac{\log_{10}(\beta_{\text{min}}/ \beta_{\text{max}})}{\log_{10}(R_{\mathrm{D}}/R_{\min,\mathrm{D}})}.$

Furthermore, the power of each D2D transmitter is updated with every transmission as
\begin{equation}\label{equ:updatepk}
p_{k}^{(t+1)}=p_{k}^{(t)}\left(\dfrac{\beta_{k}(d_{k,k}^{(t)})}{\text{SINR}_{k}(K,\textbf{p}^{(t)})}\right)^\eta,
\end{equation}
where $\eta$ is a control parameter given by $(1-\upsilon)^{-1}$~\cite{yates1997soft}. Finally, the achieved power $p_{k}^{(t+1)}$ is constrained as follows
\begin{equation*}
p_{k}^{(t+1)} = \min\lbrace P_{\max,\mathrm{D}},\max\lbrace p_{k}^{(t+1)}, P_{\min,\mathrm{D}} \rbrace \rbrace.
\end{equation*}
The SDDPC scheme is a distributed approach and the target SINR $(\beta_{k}(d_{k,k}))$ depends on the distance between the D2D pair; therefore, decision making is done by the D2D users themselves. In particular, the D2D receivers can use the sidelink control channel (e.g., Physical Sidelink Control Channel (PSCCH)) as per the LTE technical specification in 3GPP TS 36.331~\cite{3gpp331} to report back to the corresponding D2D transmitter the received SINR value and the distance based path-loss $d_{k,k}$ whenever the received SINR is below the target value.

The SDDPC scheme is summarized in Algorithm~\ref{Alg:2SDDPC}.\\

\begin{algorithm}
\caption{Dynamic Distributed Power Control}\label{Alg:2SDDPC}
\small
\begin{algorithmic}
\Procedure{SDDPC}{}
\State $ p_k^{(t)} \gets P_{\min,\mathrm{D}}= \rho_\mathrm{rx} R_{\min,\mathrm{D}}^{\alpha}(1+\varepsilon)$
\State Calculate $\beta_{k}(d_{k,k})$ according to ~(\ref{Gamma})
\If {$\text{SINR}_{k} (K,\textbf{p}) < \beta_{k}(d_{k,k})$}
\State   LOOP: \textbf{While} {$\text{SINR}_{k} (K,\textbf{p}) < \beta_{k}(d_{k,k}) $ \textbf{and} $p_{k}^{(t)} \neq P_{\max,\mathrm{D}}$} \textbf{do}
\State      \qquad\qquad $p_{k}^{(t+1)} \gets p_{k}^{(t)}\left(\dfrac{\beta_{k}(d_{k,k}^{(t)})}{\text{SINR}_{k}(K,\textbf{p}^{(t)})}\right)^\eta$
\State \qquad\qquad $p_{k}^{(t+1)} \gets \text{min}\lbrace P_{\max,\mathrm{D}},\text{max}\lbrace p_{k}^{(t+1)}, P_{\min,\mathrm{D}} \rbrace \rbrace$
\State      \qquad \qquad \textbf{goto} LOOP
\Else {\: $p_k^{(t+1)} \gets p_k^{(t)}$}
\EndIf
\State \textbf{end}
\EndProcedure
\end{algorithmic}
\end{algorithm}

%
\subsection{Discussion}\label{complexity}
On complexity and convergence of Algorithms~\ref{alg:DPPC} and~\ref{Alg:2SDDPC}, we note that Algorithm~\ref{alg:DPPC} is a non-iterative, low complexity algorithm $O(1)$, which requires around 4 simple computations. Convergence is not an issue since it is non-iterative. For Algorithm~\ref{Alg:2SDDPC}, the power allocated to the D2D users is chosen iteratively and in a non-decreasing manner. At each iteration, $p_k$ is increasing which increases $\text{SINR}_k$ until $\text{SINR}_k$ approaches the target $\beta_k$. Since the D2D TX has finite available power, the $\text{SINR}_k$ achieved by the proposed algorithm is also finite. 	For these reasons and following the same methodology as \cite{yates1997soft,yates1997soft2}, the proposed algorithm is guaranteed to converge to a finite $\text{SINR}_k$. The proof is similar to Theorem 3 in \cite{yates1997soft,yates1997soft2} and hence is omitted for brevity. Furthermore, figure ~\ref{fig:Iterations_SDDPC} shows the number of iterations needed in this algorithm that are very low. For instance, as $M$ increases, the number of D2D links $K'$, sharing the resources with one of the cellular users, decreases; therefore the interference level caused by the D2D users will decrease and hence increasing the $\text{SINR}_k$. This will cause Algorithm~\ref{Alg:2SDDPC} (SDDPC) to converge faster (for $M=3$, it requires an average of 3 iterations to converge).

\begin{figure}[t]
    \centering
    \includegraphics[scale=0.85]{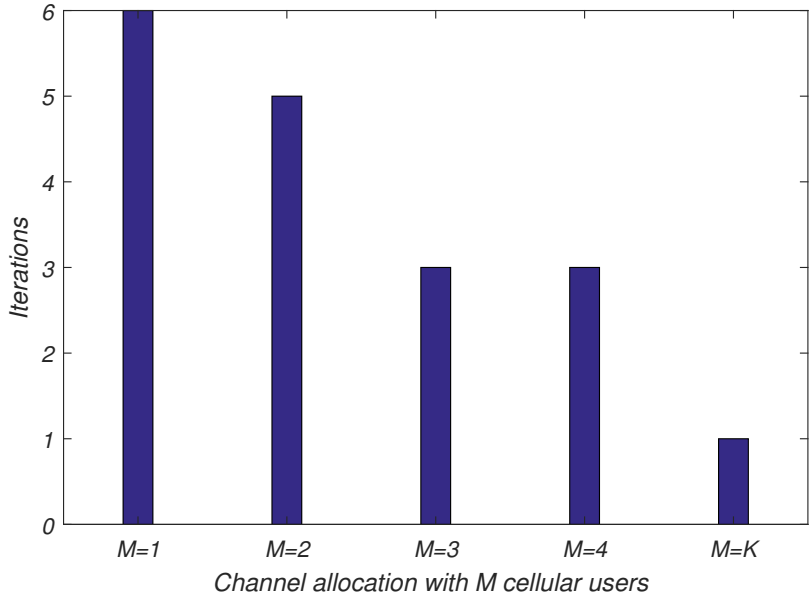}
    \caption{Number of iterations for the SDDPC scheme for $\lambda=5\times 10^{-5}$ and different $M$ channel allocations. }
\label{fig:Iterations_SDDPC}
\end{figure}

Moreover, Algorithms~\ref{alg:DPPC} and~\ref{Alg:2SDDPC} may not necessarily converge to the global optimal solutions. The development of global optimal power allocation is otherwise done in a centralized manner at the base station. However, it would require excessive signaling overhead in which the computational complexity grows exponentially with $K$ ~\cite{lee2015power,memmi2016power}. This excessive overhead is avoided in the distributed case, with graceful degradation in performance.

Furthermore, we note that using the two proposed static distributed PC schemes for link establishment, the allocated power remains constant over the resource blocks since we apply equal power allocation to all the assigned resource blocks. On the other hand, for link maintenance, SDDPC compensates the measured SINR at the receiver with a variable target SINR. The power allocated per PRB of each D2D UE is updated every transmission as per~\eqref{equ:updatepk}.

In order to realize the proposed PC schemes, each D2D transmitter needs to have knowledge of: 1) the distance based path-loss parameters $d_{k,k}^{\alpha}$ and $d_{0,k}^{\alpha} $ in order to allocate power, 2) the target SINR $\beta$, 3) the density of the D2D links $q\lambda$, and (4) CSI of the direct link. Knowledge of distance based path-loss $d_{k,k}^{\alpha}$ and $\beta$ can be acquired through feedback from the corresponding D2D receiver. During D2D link establishment~\cite{3gpp803}, the density of the D2D links (which is the average number of active D2D links per unit area) as well as $d_{0,k}^{\alpha}$ can be estimated at the eNB. The D2D transmitters acquire the density $q\lambda$ when the eNB broadcasts it using the downlink control channel, and acquire $d_{0,k}^{\alpha}$ through feedback from the eNB.

All D2D pairs can use the sidelink channels (Physical Sidelink Broadcast Channel (PSBCH) and PSCCH)~\cite{3gpp331} to transmit reference signals to enable D2D receivers to perform measurements and report them back to the eNB or to the corresponding D2D transmitter. Each D2D receiver can reliably estimate the distance based path-loss parameters using these signals by averaging the effects of fading over multiple resource blocks.

The eNB can also estimate distances through the location updates defined in 3GPP TS 23.303~\cite{3gpp303}, and the path-loss exponent can be estimated as per \cite{chandrasekharan2015propagation} through defining path-loss exponents based on the region of the D2D pairs location. The UE's location information exchanged is expressed in shapes as defined in 3GPP TS 23.032~\cite{3gpp032} as universal geographical area description (GAD).

%
\section{Simulation Results}\label{sec:simulation}
In this section, we provide numerical results for the D2D underlaid cellular network. First, we show how the estimation error margin ($\varepsilon$) and the PC control parameter ($\mu$) for DPPC and EDPPC affect the coverage probability for the cellular and the D2D links. Then, we show the performance gains of using the proposed CA and PC schemes (compared to the on/off PC in \cite{lee2015power}) in terms of coverage probability, spectral and energy efficiency.

%
\subsection{Simulation Setup}
Figure~\ref{fig:Model_sparse} shows a snap shot depicting the geometry of a typical cell. The eNB is located at the center position $(0,0)$ and the uplink users are uniformly located within a radius $\RC$. The D2D transmitters are located according to a PPP distribution with $\lambda \in \left\lbrace 2\times10^{-5},5 \times10^{-5} \right\rbrace $ in a ball centered at the eNB and radius $\RC+\unit[250]{ m}$. The system parameters used throughout the experimental simulations are summarized in Table \ref{table:1}. Moreover, the transmit power of the cellular user is set as $p_0 = P_{\max,\mathrm{C}}$.
\begin{table}[t!]
\centering
\caption{Simulation Parameters}
\label{table:1}
\begin{tabular}{l|c}
 \hline
 \textbf{Parameter} & Value   \\\hline
    Cell radius ($\RC$) & $\unit[500]{m}$   \\\hline
    Max. D2D link range ($R_{\mathrm{D}}$)  &  $\unit[50]{m}$ \\\hline
    Min. D2D link range ($R_{\min,\mathrm{D}}$)  &  $\unit[5]{m}$  \\\hline
    D2D link density ($\lambda$) & $2 \times10^{-5}$ and $5 \times10^{-5}$ \\\hline
    Average $\#$ D2D links ($K$) & $\mathbb{E}\left[K \right] = \pi\lambda \RC^2 \in \left\lbrace 15, 39\right\rbrace $ \\\hline
    Path-loss exponent ($\alpha$) & 4 \\\hline
    Target SINR threshold ($\beta$) & varies from $\unit[-18]{dB}$ to $\unit[18]{dB}$ \\\hline
    Max. TX power of cellular user \cite{memmi2016power}& $P_{\max,\mathrm{C}} =  \unit[100]{mW}$ \\\hline
    Max. TX power of D2D user \cite{lee2015power} & $P_{\max,\mathrm{D}} =  \unit[0.1]{mW}$ \\\hline
    Min. TX power of D2D user & $P_{\min,\mathrm{D}} =  \unit[0.2]{\mu W}$ \\\hline
    Estimation margin $\varepsilon$ & 0.5 \\\hline
	Channel quality threshold for EDPPC $G_{\min}$ &  $\unit[-40]{dbm}$ \\\hline
	PC parameter $\mu$ for EDPPC & 0.0005\\\hline
    Receiver sensitivity $\rho_\mathrm{rx}$ & $\rho_\mathrm{rx}=P_{\max,\mathrm{D}} {R_{\mathrm{D}}}^{-\alpha} $\\\hline
    Noise variance ($\sigma^2$) &  $\unit[-112.4]{dBm}$    \\\hline
    Monte-Carlo Simulations & 1000\\\hline
    TTI &  $\unit[1]{ms}$  \\\hline
\end{tabular}
\end{table}

\begin{figure*}[t!]
\centering
\subfigure[]{
\includegraphics[scale=0.45]{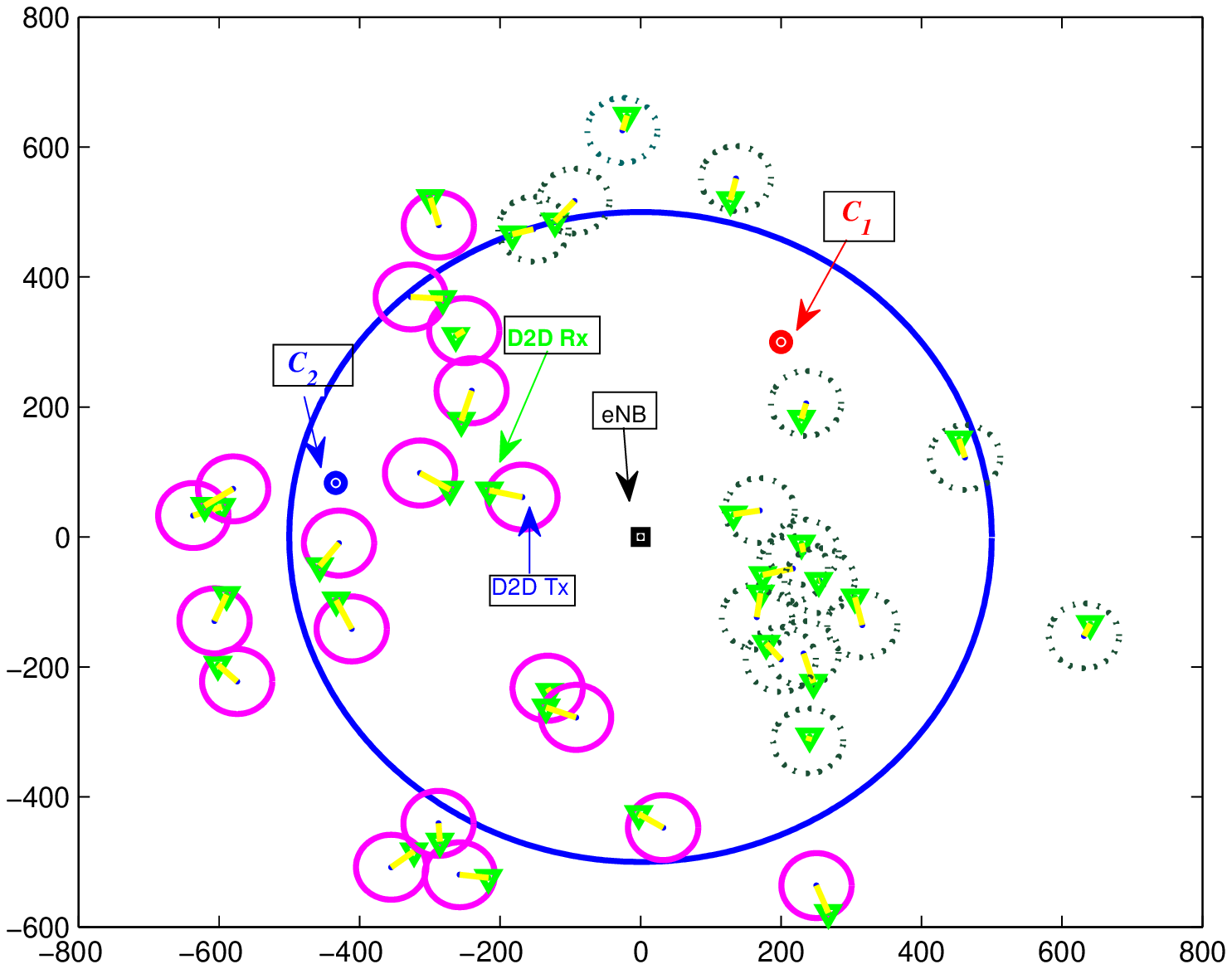}
\label{fig:Model_sparse}}\hspace{-0.24in}\vspace{-0.1in}
\subfigure[]{
\includegraphics[scale=0.5]{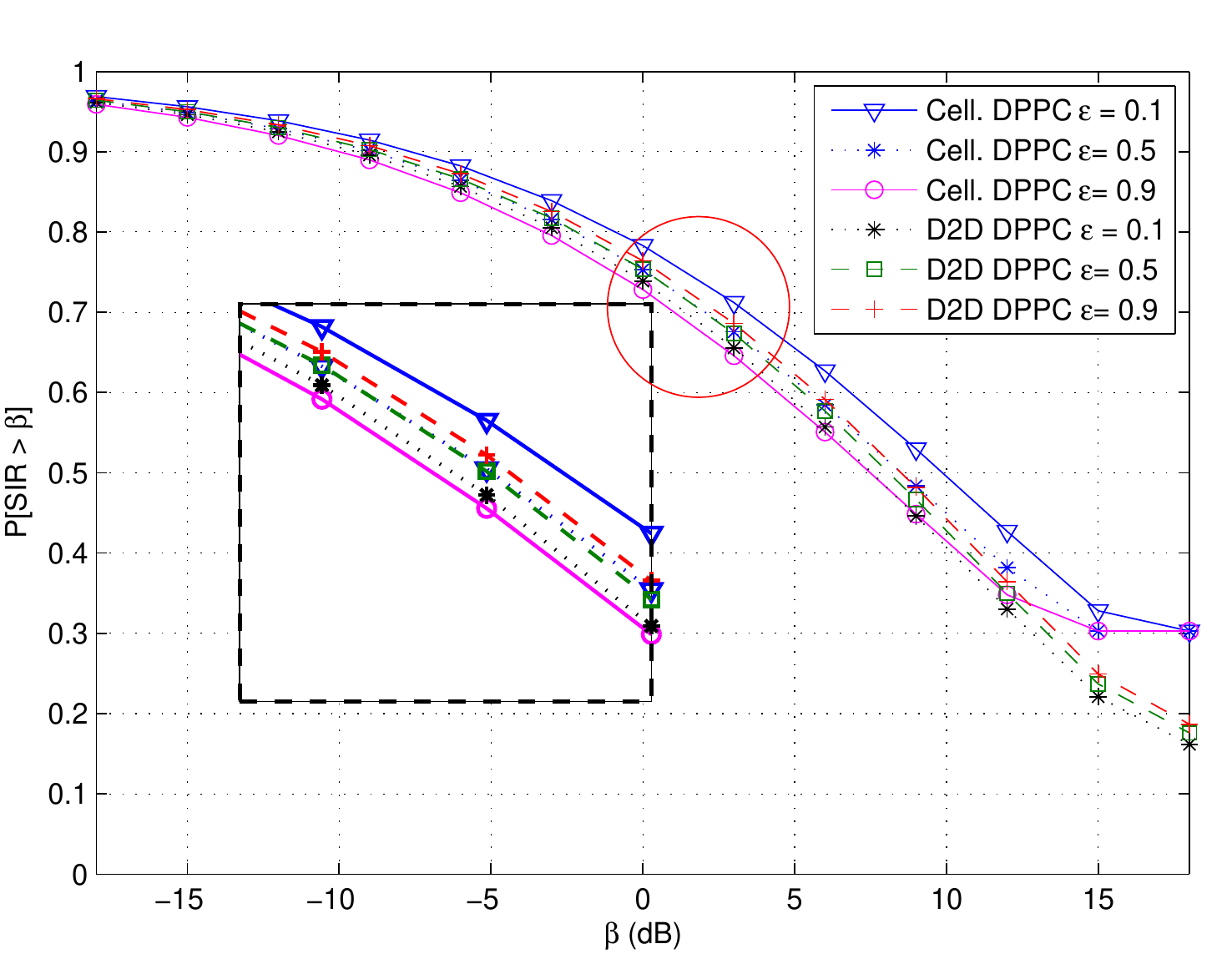}
\label{fig:DPPC_variable_epsilon2}}
\quad
\subfigure[]{
\includegraphics[scale=0.5]{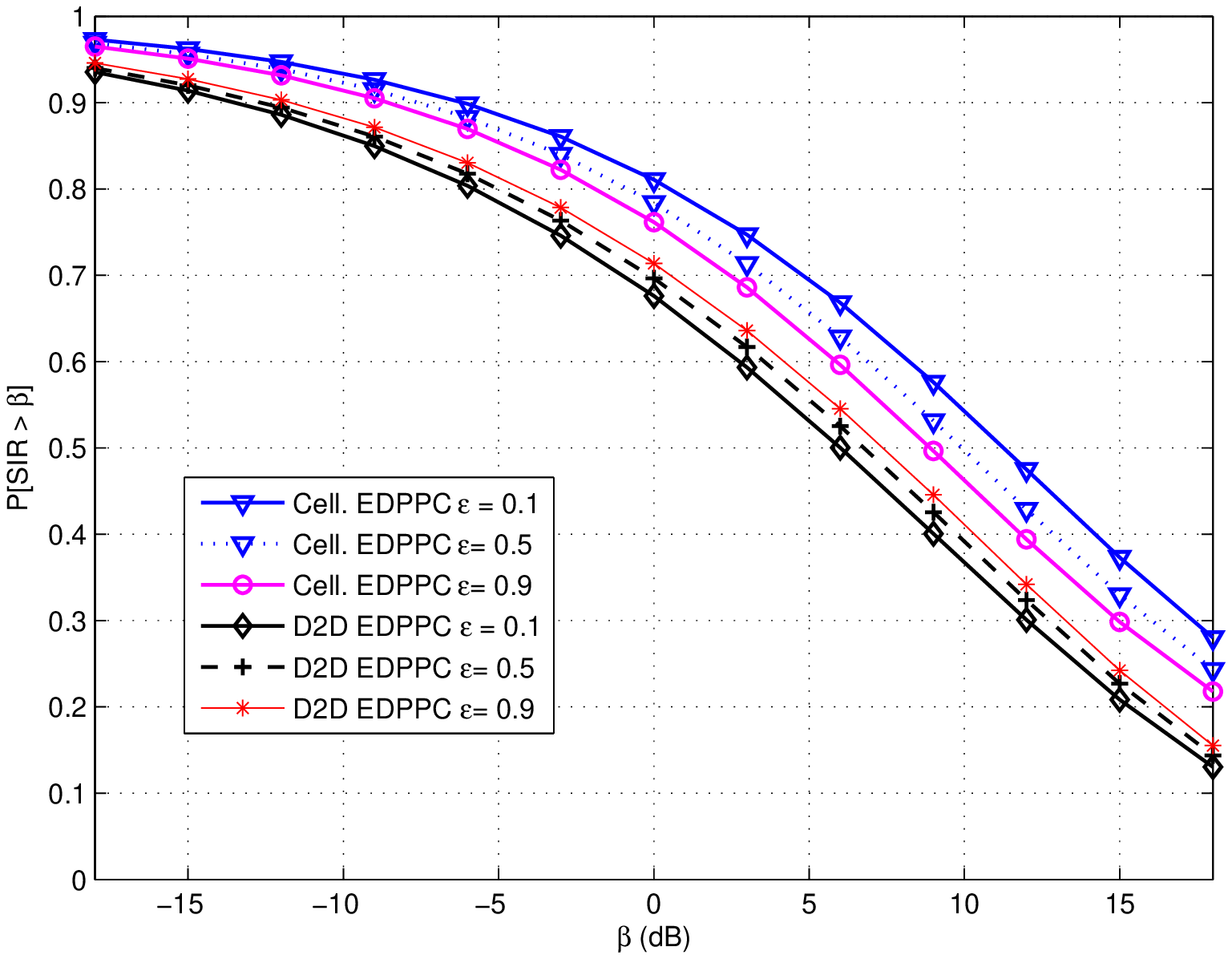}
\label{fig:EDPPC_variable_epsilon}}\hspace{-0.35in}\vspace{-0.1in}
\quad
\subfigure[]{
\includegraphics[scale=0.52]{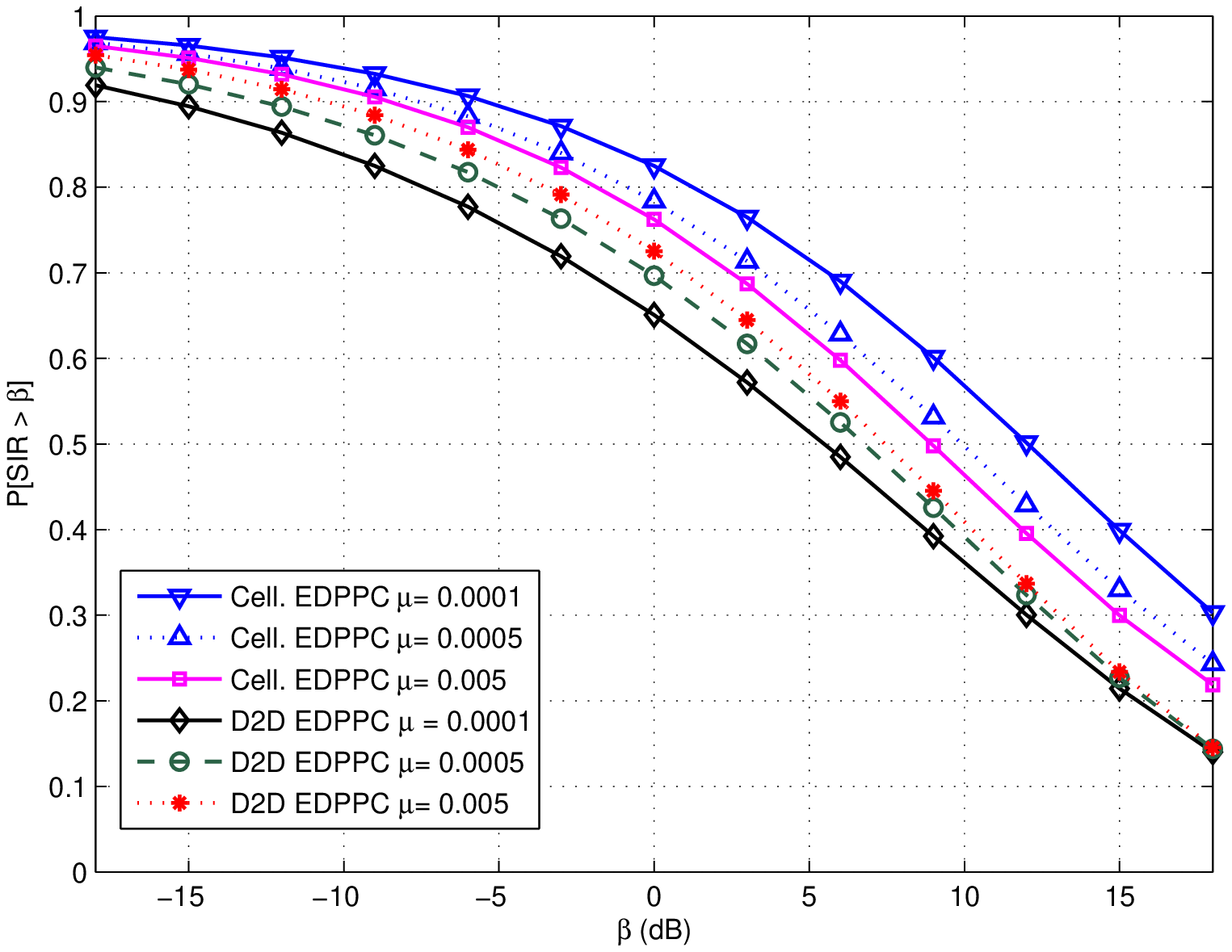}
\label{fig:EDPPC_variable_mu}}
\caption{(a) A snapshot of link geometry for a D2D underlaid cellular network assuming a sparse D2D link deployment scenario (i,e., $\lambda = 2 \times10^{-5}$). D2D links in circles share resources with CUE $c_1$, while D2D links in dashed circles share resources with CUE $c_2$. (b) Coverage probability for cellular and D2D users where resources are shared with 2 CUEs, using the proposed DPPC with variable $\varepsilon$. (c) Same as (b) but using the proposed EDPPC scheme. (d) Using the proposed EDPPC with variable $\mu$.}
\label{fig:globfig}
\end{figure*}

\begin{figure*}[t!]
\centering
\subfigure[]{
\includegraphics[scale=0.6]{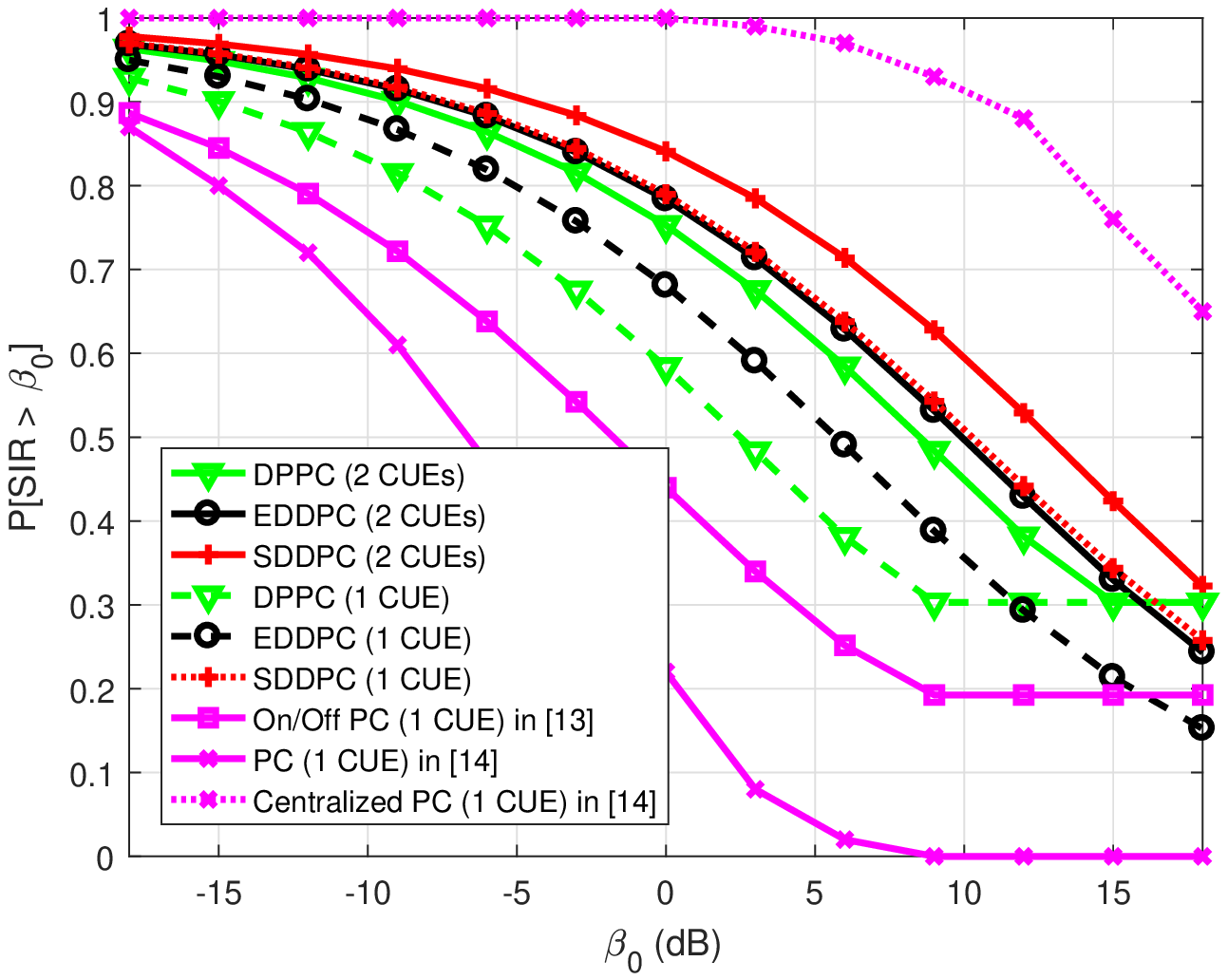}
\label{fig:Cell_Dense_All}}\hspace{-0.2in}\vspace{-0.1in}
\quad
\subfigure[]{
\includegraphics[scale=0.6]{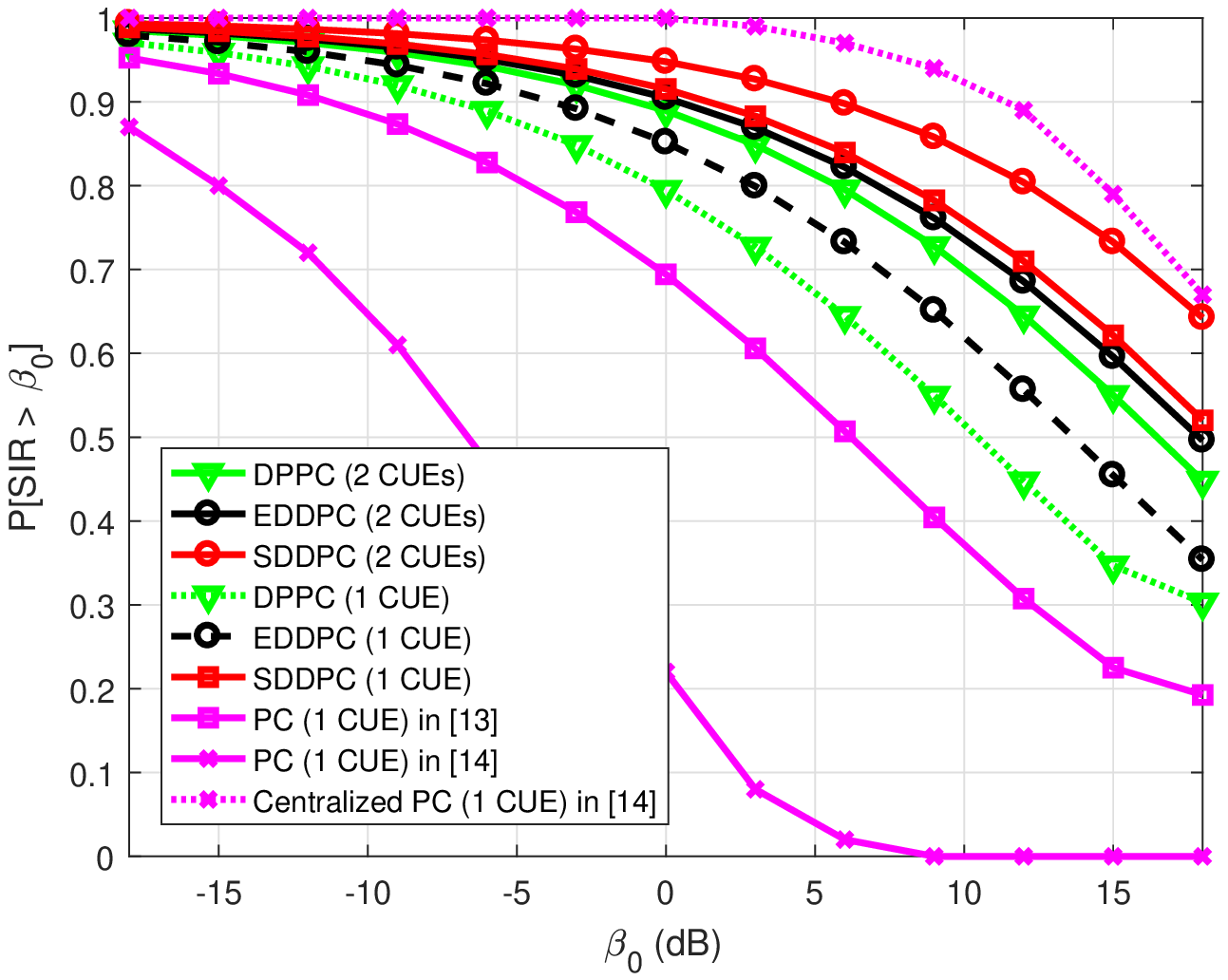}
\label{fig:Cell_Sparse_All}}
\bigskip
\subfigure[]{
\includegraphics[scale=0.6]{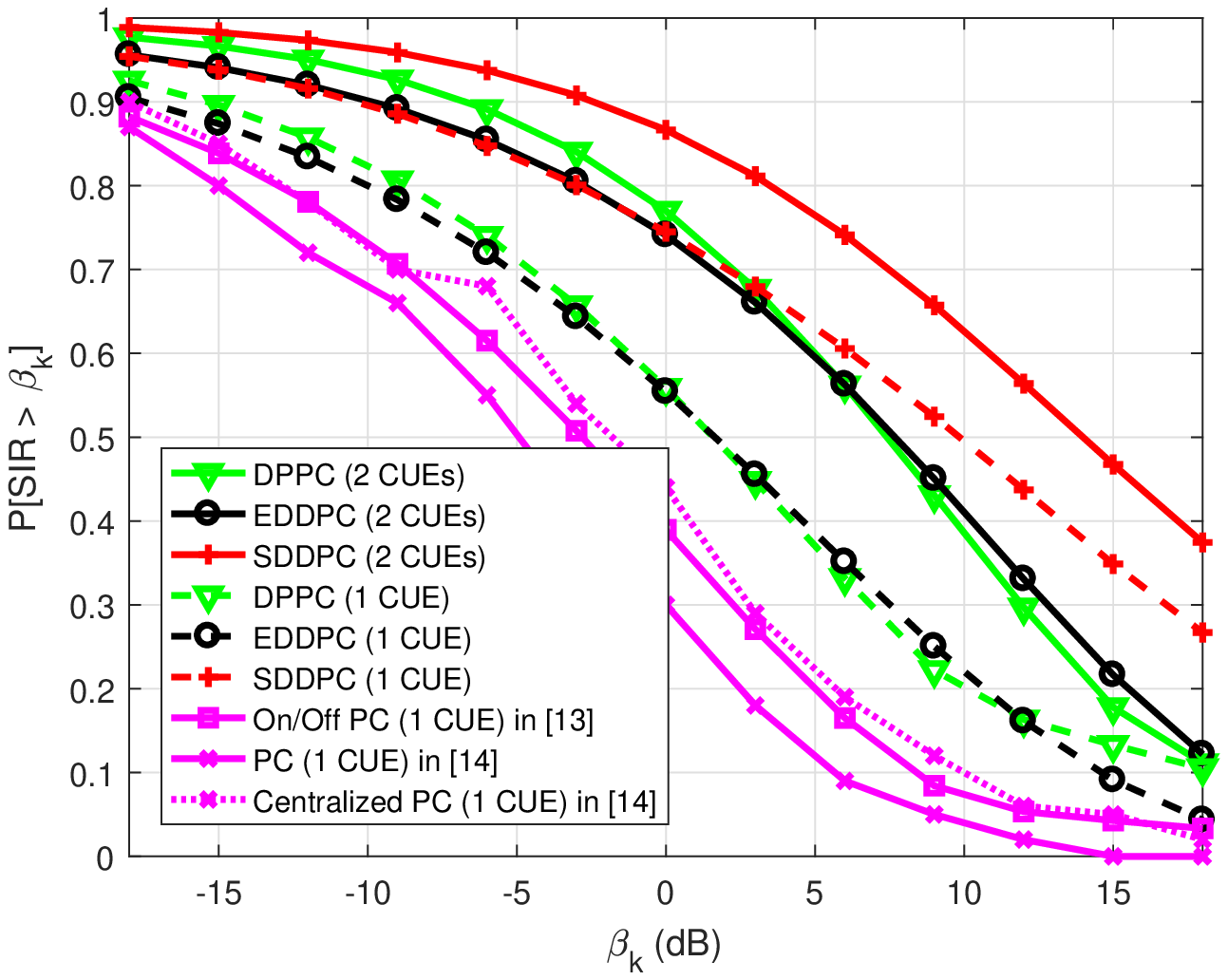}
\label{fig:D2D_Dense_All}}\vspace{-0.3in}\hspace{-0.2in}\vspace{-0.05in}
\quad
\subfigure[]{
\includegraphics[scale=0.6]{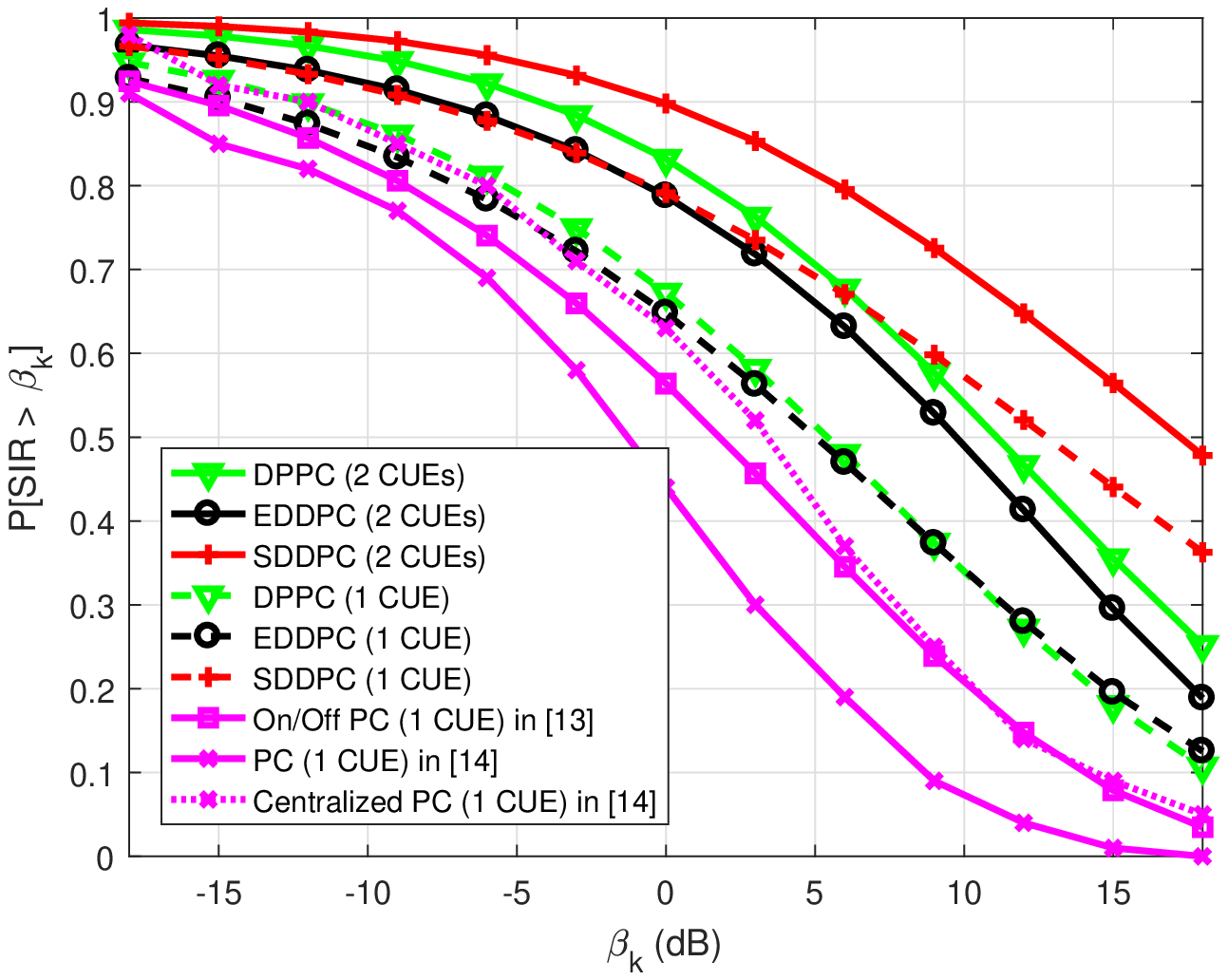}
\label{fig:D2D_Sparse_All}}
\caption{Coverage probability for cellular and D2D users using all the proposed PC schemes in this work vs. that of \cite{lee2015power,memmi2016power}: (a) For cellular users in dense network scenario, (b) for cellular users in sparse network scenario, (c) for D2D users in dense network scenario, and (d) for D2D users in sparse network scenario.}
\label{fig:globfig2}
\end{figure*}

%
\subsection{Coverage Probability for DPPC and EDPPC with Variable Parameters}
In a dense D2D link deployment scenario, the average number of D2D links in the cell is $\mathbb{E}\left[K \right]=39$  and the average number of D2D links sharing resources with one of the two cellular users is $\mathbb{E}\left[K' \right]=20$. For the case of variable $\varepsilon$ for both DPPC and EDPPC, we plot the cellular and D2D coverage probability in Figs.~\ref{fig:DPPC_variable_epsilon2} and ~\ref{fig:EDPPC_variable_epsilon}. As shown in the figures, as the error margin varies from 0.1 to 0.9 ($\mu=0.0005$), the cellular coverage probability decreases while the D2D coverage probability increases. D2D users allocate more power to enhance the D2D link, thus causing more interference to the cellular users. In addition, it is noted in Fig.~\ref{fig:DPPC_variable_epsilon2} for DPPC that no D2D link is dropped when $\beta< \unit[10]{dB}$, since the transmit probability $\txprob=\min\left\{ \tfrac{2\text{sinc}(\: 2/\alpha\:)}{\pi q\lambda\: \beta^{2/\alpha} \: {R_{\mathrm{D}}^2}},1\right\}=1$. However, when $\beta> \unit[10]{dB}$, the transmit probability is activated where $\txprob\neq 1 \vspace{2pt}$, and some D2D links are dropped thus reducing the D2D interference and enhancing the link coverage for D2D and cellular transmitters.

In Fig.~\ref{fig:EDPPC_variable_mu}, we vary the control parameter $\mu$ for $\varepsilon=0.5$ using the EDPPC scheme. As $\mu$ decreases from 0.005 down to 0.0001, the cellular coverage probability increases and D2D coverage probability decreases. Hence D2D links are dropped according to $\mu$ so that they do not cause excessive interference to cellular users. Furthermore, the remaining D2D users will allocate less power, thus the interference at the cellular users and at the other D2D users will be even more diminished. Therefore, the proposed scheme can effectively protect cellular users from interference caused by the D2D users.

%
\subsection{Cellular Coverage Probability for all PC schemes}
In Figs. ~\ref{fig:Cell_Dense_All} and ~\ref{fig:Cell_Sparse_All}, we plot the coverage probability of the cellular links using our proposed schemes for two scenarios where the D2D links share the resources with one and two cellular users in dense and sparse networks. We also compare the results with that of 1) the on/off PC scheme in~\cite{lee2015power}, which are the same results as in~\cite{memmi2016power} for the best case scenario with zero channel uncertainty, and 2) the on/off PC scheme in~\cite{memmi2016power} with channel uncertainty factor of 0.5. It can be seen that all the proposed schemes outperform the scheme in \cite{lee2015power,memmi2016power}. In particular for the case of 2 CUEs, SDDPC increases the coverage probability by more than $40 \%~ (45\%)$ in dense (sparse) networks compared to~\cite{lee2015power,memmi2016power} for the entire range of $\beta_0$. The EDPPC scheme performs better than DPPC due to the extra $d_{0,k}^{-\alpha}$ parameter that further reduces the interference at the eNB. However, SDDPC outperforms the other PC schemes as it protects the cellular links using the adaptive approach.

As expected, the cellular coverage probability increases when D2D users share resources with multiple cellular users. The reason is that a smaller number of D2D links share the same resources with a particular CUE, which results in a reduction in the interference caused by the D2D transmissions.

In addition, one can note that the centralized power control \cite{memmi2016power} achieves nearly perfect cellular user coverage probability performance in the low target SINR values, at high cost of system complexity as discussed in Section.~\ref{complexity}.

%
\subsection{D2D Coverage Probability for all PC schemes}
Figures~\ref{fig:D2D_Dense_All} and~\ref{fig:D2D_Sparse_All} show the coverage probability of D2D links using the proposed PC schemes in dense and sparse network deployments. As shown, all proposed schemes outperform the schemes in \cite{lee2015power,memmi2016power}. On one hand, the coverage probability for SDDPC increases by up to $60\%~(50\%)$ for the dense (sparse) scenario. On the other hand, DPPC and EDPPC have approximately similar performance where the coverage probability increases by $40\%~ (30\%)$. However, SDDPC outperforms the other PC schemes, since the D2D links set variable target SINRs. For instance, links with good quality have high SINR target, while links with low quality have low SINR target.

Moreover, when D2D users share resources with more than one cellular user, the D2D coverage probability using our proposed PC schemes is significantly enhanced as the interference caused by the D2D transmission on other D2D users is reduced.

In general, the D2D coverage probability performance decreases in the dense scenario; however, the total number of successful D2D transmissions is larger than that of the sparse D2D link deployment scenario. For instance, when the target SINR is $\unit[0]{dB}$, the total number of successful D2D transmissions in both sparse and dense scenarios is $|A_\mathrm{D}|_{\text{sparse}}=\mathbb{E}\left[K \covprobD{\beta_k}\right] = 15\times 0.9 \approx 13$ and $|A_\mathrm{D}|_{\text{dense}}= 39\times 0.88 \approx 34$, respectively, using the proposed SDDPC scheme and resources are shared with 2 CUEs. The corresponding numbers of successful D2D transmissions from~\cite{lee2015power} are $|A_\mathrm{D}|_{\text{sparse}}= 15\times 0.58 \approx 8$ and $|A_\mathrm{D}|_{\text{dense}}= 39\times 0.4 \approx 15$, respectively. Therefore, a significant increase in the number of the D2D links is attained using the proposed SDDPC scheme.

In addition, one can note that the centralized power control \cite{memmi2016power} (with high signaling overhead and complexity) does not perfrom as well for the D2D case, since this approach works on maximizing the SINR of the uplink user and allows less D2D links to access the network through the admission control. However, with less complexity, our proposed schemes outperform the centralized approach.
\begin{figure*}[t]
\centering
\subfigure[]{
    \includegraphics[scale=0.55]{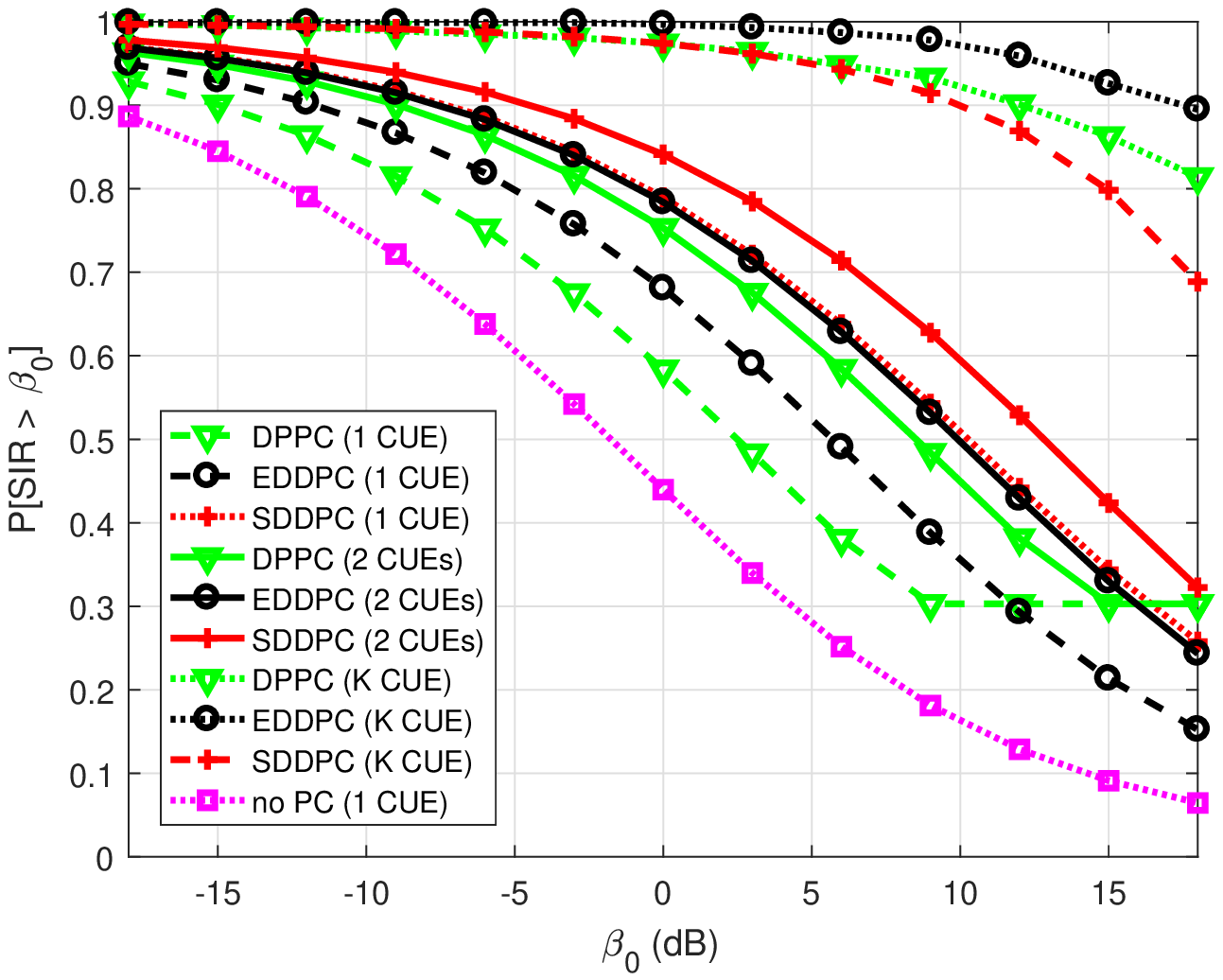}
	\label{fig:Cell_dense_ALL_M}}
\quad
	\subfigure[]{
	\includegraphics[scale=0.55]{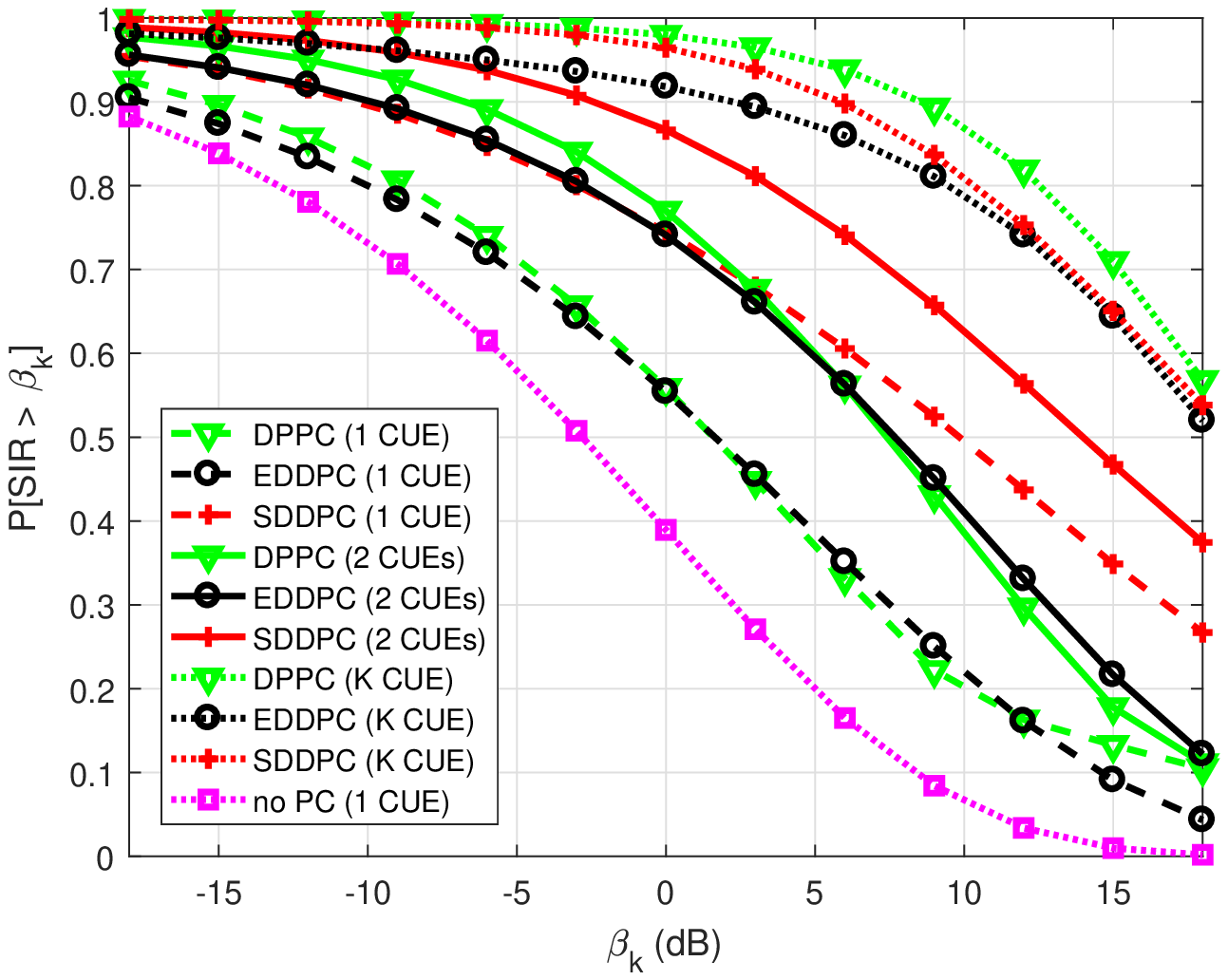}
\label{fig:D2D_dense_ALL_M}}
\caption{Coverage probability for: (a) cellular, and (b) D2D users, for  $M=1, M=2, \text{ and } M=K$.}
\label{fig:MCellfig}
\end{figure*}

\subsection{Coverage Probability with Variable Channel Allocation parameter ($M$)}
Figures~\ref{fig:Cell_dense_ALL_M} and~\ref{fig:D2D_dense_ALL_M} show the coverage probability of cellular  and D2D links using the proposed PC schemes in dense network while varying the channel allocation parameter $M$. Upon increasing $M$, the coverage probability for the D2D and cellular users is enhancing, since a smaller number of D2D users (which share the same resources) will generate interference. Moreover, we have considered the maximum allocation case where $M=K$ in which one cellular uplink will share the resources with only one D2D link and $\expec{K'}=\tfrac{\mathbb{E}[K]}{M}=1$. In this case, the uplink signal will observe interfernce from only the farthest D2D user, and the D2D link will observe the interference from only the farthest cellular uplink user. Thus, the coverage probability for the D2D and the cellular link is greatly enhanced. Furthermore, we compare our results with the case of no power control applied at the D2D links where $p_k=P_{\max,\mathrm{D}}$ and, as expected, the coverage probabilities are drastically affected (decreased by more than $20\%$); since the D2D-interference is overwhelming the receivers (Base station and the D2D receivers). Therefore, our proposed channel allocation and power control schemes are effective interference mitigation methods in order to guarantee the QoS of the cellular uplinks and D2D links.

%
\subsection{Spectral and Power Efficiency}
\begin{figure}[t!]
\centering
    \includegraphics[width=1\textwidth]{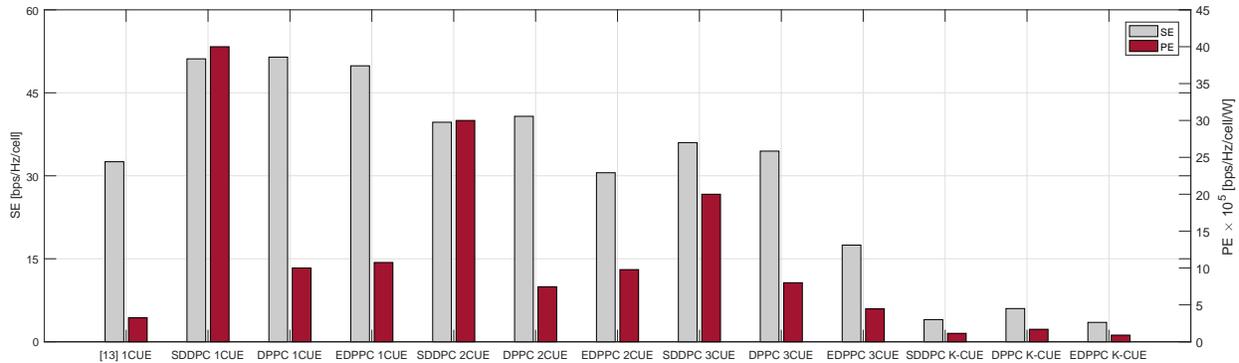}
    \caption{Spectral efficiency (left axis) and power efficiency (right axis) for the proposed PC schemes in a dense network.}
\label{fig:SE_PE_Dense}
\end{figure}

Figure~\ref{fig:SE_PE_Dense} shows the spectral and power efficiency of the D2D and cellular system when applying the proposed PC schemes in a dense deployment scenario, where resources are shared with one, two, and three cellular users. The spectral and power efficiency are defined as follows: 1) Spectral efficiency (SE) is the sum rate $\mathcal{R}_s^{(D)}$ for all D2D links in $\unit[]{bps/Hz/cell}$ as defined in ~(\ref{equ:rate}), and 2) Power efficiency (PE) is given by the ratio of the D2D spectral efficiency achieved over the average transmit power of the D2D links in $\unit[]{bps/Hz/cell/W}$. The figure shows that SDDPC is spectrally and power efficient since more D2D links are able to achieve higher SINR values and less power is allocated for the D2D links. For the case of static PC schemes, DPPC is more spectrally efficient because it maximizes the D2D sum rate. On the other hand EDPPC is more power efficient than DPPC, since less power is allocated. In addition, as expected, when sharing resources with more cellular users, the spectral efficiency of the system decreases (by $15\%$ as shown); however, the coverage probabilities (for cellular and D2D) increase because the interference level is reduced, as shown in Fig.~\ref{fig:globfig2} and Fig.~\ref{fig:MCellfig}.

In addition, when considering $M\!=\!2$ the performance is efficient in the sense that it gives a compromise between coverage, spectral efficiency, and complexity. The eNB performs only one comparison for each active D2D ($d_{k,c_1} \!>\! d_{k,c_2}$). However, when $M$ is further increased, the spectral efficiency for the EDDPC becomes lower than that in~\cite{lee2015power} and the complexity increases.

A trade-off exists between spectral efficiency, power efficiency, and coverage probability. If it is desired for the cellular and D2D link to be of high quality, then CA should be applied with D2D users sharing resources with more than one cellular user. However, if it is required that the D2D communications to be power efficient and not cause too much interference on the uplink, then EDPPC has an advantage over the DPPC. Otherwise, if spectral efficiency is a priority, then DPPC performs best, particularly when D2D users share resources with one cellular user. Finally, SDDPC proves most adequate for link maintenance since it is both spectrally and power efficient, and it maintains the link quality for both D2D and cellular users.

%
\section{Conclusion}\label{sec:conclusion}
In this paper, a random network model for a D2D underlaid cellular system based on stochastic geometry has been proposed. Using this modeling approach, it is possible to derive closed-form analytical expressions for the coverage probabilities and ergodic sum-rates, which give insight into how the various network parameters interact and affect link performance and quality. Unlike previous work, it is shown that a channel allocation scheme that allows D2D links to share resources with more than just one cellular user has merit. New power control schemes targeted for D2D link establishment and link maintenance have been shown to adequately control interference levels under various static and dynamic conditions, using distance-based path-loss parameters (with error margin), varying target SINR, and local CSI. It has been shown through experimental simulations that network performance in terms of coverage probability and spectral efficiency is improved by activating more underlaid D2D links while maintaining the quality of cellular links, and at the same time enhancing power efficiency.

Extending the proposed schemes to include the effects of out-of-cell interference and multiple cell scenarios will be examined in a future work.

%
\section*{Appendix}

%
\subsection{Proof of Lemma~\ref{lemma1}}\label{proof:lemma1}
For the case of two cellular users, we have using~\eqref{equdkc}:
\begin{equation}\label{equthin_lambda}
\begin{split}
\mathbb{P}[Q_k=1]=&\mathbb{P}[\overbrace{d_{k,c_1}}^x \geq \overbrace{d_{k,c_2}}^y] =\iint\limits_{(x,y; x\geq y)} f_{X,Y}(x,y) dx\: dy = \int_{0}^{2R_\text{C}}\!\int_{0}^{x} f_{Y}(y)dy \: f_{X}(x) dx \\
&=\int_{0}^{2R_\text{C}}\!\!\underbrace{\int_{0}^{x}\!{\tfrac{2y}{R_\text{C}^{2}}}\!\left(\!{\tfrac{2}{\pi}}\cos^{-1}\!\left({\tfrac{y}{2R_\text{C}}} \right)\!\!-\!{\tfrac{y}{\pi R_\text{C}}}\!\sqrt{1\!-\!{\tfrac{y^{2}}{4R_\text{C}^{2}}}}\right) dy}_{\mathcal{A}(x)} \: f_{X}(x) dx.
\end{split}
\end{equation}
To solve ~(\ref{equthin_lambda}), integral $\mathcal{A}(x)$ can be directly computed as follows
\begin{align*}
\mathcal{A}(x)=1+ \tfrac{2}{\pi}\left(\tfrac{x^2}{R_\text{C}^2}-1 \right)\cos^{-1}\left( \tfrac{x}{2R_\text{C}}\right) \!\!-\!{\tfrac{x}{\pi R_\text{C}}}\left(1+ \tfrac{x^2}{2R_\text{C}} \right)  \!\sqrt{1\!-\!{\tfrac{x^{2}}{4R_\text{C}^{2}}}}.
\end{align*}
Now using this expression for $\mathcal{A}(x)$, we solve~\eqref{equthin_lambda}
\begin{equation}\label{equthin_lambda2}
\mathbb{P}[d_{k,c_1} \geq d_{k,c_2}] = \int_{0}^{2R_\text{C}} \mathcal{A}(x){\tfrac{2x}{R_\text{C}^{2}}}\!\left(\!{\tfrac{2}{\pi}}\cos^{-1}\!\left({\tfrac{x}{2R_\text{C}}} \right)\!\!-\!{\tfrac{x}{\pi R_\text{C}}}\!\sqrt{1\!-\!{\tfrac{x^{2}}{4R_\text{C}^{2}}}}\right) dx={1\over 2}.
\end{equation}

For the general case of $M$ cellular users, we have:
\begin{align*}
\mathbb{P}[\overbrace{d_{k,c_1}}^{y_1} \geq \max \{ \overbrace{d_{k,c_2}}^{y_2},\cdots ,\overbrace{d_{k,c_M^{}}}^{y_M} \}]
    &=  \!\idotsint\limits_{(y_1,y_i; y_1\geq y_i)}^{M}\!\!\!f_{Y_1,Y_2}(y_1,y_2) \cdots  f_{Y_1,Y_M^{}}(y_1,y_M^{}) dy_1 \dots dy_M^{}\\
    &=  \dfrac{1}{M}
\end{align*}

%
\subsection{Proof of Proposition~\ref{Cell_prob_cov}}\label{proof:Cell_prob}
Using~\eqref{eq:cov_probC_def}, the cellular coverage probability is given by
\begin{align*}
    \covprobCC
        &=\expec{\mathbb{P} \left( {|h_{0,c_m}|^2} \geq \beta_0 d_{0,c_m}^{\alpha} p_0^{-1} \left({\sum_{x_k\in \Phi '}|h_{0,k}|^2d_{0,k}^{-\alpha} p_k + \sigma^2}\right) \right)} \\
        &=\expec{  e^{- \beta_0 d_{0,c_m}^{\alpha} p_0^{-1} \left({\sum_{x_k\in \Phi '}|h_{0,k}|^2d_{0,k}^{-\alpha} p_k + \sigma^2}\right)} } \\
        &=\expec{e^{- \beta_0 \sigma^2 d_{0,c_m}^{\alpha} p_0^{-1} }}\expec{ e^{- \beta_0 d_{0,c_m}^{\alpha} p_0^{-1} \left({\sum_{x_k\in \Phi '}|h_{0,k}|^2d_{0,k}^{-\alpha} p_k}\right)} }
\end{align*}
For the proposed channel allocation scheme, the Laplace transform $\mathcal{L}_{\Phi'}(s)$ is given as
\begin{equation*}
\mathcal{L}_{\Phi '}(s)\triangleq \expec{  e^{- s \left({\sum_{k\in \Phi '}|h_{0,k}|^2d_{0,k}^{-\alpha} p_k}\right)}_{} }=e^{-\tfrac{\pi}{\text{sinc}(2/\alpha)}  \expec{p_k^{2/\alpha}} \mathbb{P}[Q_k=1] \lambda s^{2/\alpha}}
\end{equation*}
where $s = \beta_0 d_{0,c_m}^{\alpha} p_0^{-1}$. The result follows.
Furthermore, it turns out that the Laplace transform is easier than determining the distribution functions, and it completely characterizes the distribution of PPP \cite{baccelli2009stochastic,schilcher2016interference}.

\subsection{Proof of Proposition~\ref{D2D_prob_cov}}\label{proof:D2D_prob}
We first need to derive the expectation of the interference term from other D2D users. Using Slivnyak's theorem~\cite{baccelli2009stochastic} and considering the proposed channel allocation scheme, the reduced PPP excluding the $\nth{k}$ point ($\Phi'\setminus{x_k}$) remains the same as the original PPP $\Phi'$. Hence,
\begin{align*}
&\mathcal{L}_{\Phi'\setminus\{x_k\}} (s)=\expec{e^{-s{\sum_{x_i\in\Phi'\setminus\{x_k\}}p_{i}\vert h_{k,i}\vert^{2}\Vert x_{i}\Vert^{-\alpha}}} \vert k\in\Phi'}=\,{\mathcal L}_{\Phi'}(s)=e^{-{\pi\mathbb{P}[Q_k=1]\lambda\over{\rm sinc}\left({2/\alpha}\right)} \expec{p_{k}^{2/\alpha}}s^{2/\alpha}}.
\end{align*}
Therefore, the coverage probability of the D2D links is given by
\begin{align*}
\mathbb{P}(\text{SINR}_k \geq \beta_k)
    &=\!\mathbb{P}\!\left(\!\vert h_{k,k}\vert^{2}\!\geq\!{\beta_k d_{k,k}^{\alpha}\over p_{k}} \!\left(\!\sum_{x_i\in\Phi'\setminus\{x_k\}}\!{p_{i}\vert h_{k,i}\vert^{2}\over\Vert x_{i}\Vert^{\alpha}}\!+\!{p_{0} \vert h_{k,c_m}\vert^{2}\over d_{k,c_m}^{\alpha}}+\sigma^{2}\right)\!\right) \\
    &=\!\expec{e^{-\beta_k p_{k}^{-1} d_{k,k}^{\alpha}\!\left(\sum_{x_i\in\Phi'\setminus\{x_k\}}\!p_{i}\vert h_{k,i}\vert^{2}\Vert x_{i}\Vert^{-\alpha} + p_{0} \vert h_{k,c_m}\vert^{2}d_{k,c_m}^{-\alpha}+\sigma^{2}\right)}} \\
    &=\mathbb{E}_{Z} \left[e^{-\sigma^{2}\beta p_{k}^{-1}d_{k,k}^{\alpha}}{\cal L}_{\Phi'}\left(\beta_k p_{k}^{-1}d_{k,k}^{\alpha} \right){\cal L}_{Y}\left(\beta_k p_{k}^{-1}d_{k,k}^{\alpha}\right)\right],
\end{align*}
where $Z=d_{k,k}^{\alpha}p_{k}^{-1}$, $Y =|h_{k,c_m}|^2d_{k,c_m}^{-\alpha} p_0$, and ${\mathcal{L}_Y(\beta_k Z)= \mathbb{E}_Y[ e^{-(\beta_k Z)Y}]}$ .

\subsection{Proof of Corollary~\ref{cor:expec_dk}}\label{proof:d_k,cm}
For simplicity, we derive the expressions for $d_{k,c_1}$. The same approach can be used for $d_{k,c_2}$.
We set $D_{c_1}=d_{k,c_1}$ as the distance from any $\nth{k}$ D2D transmitter to the cellular UE $c_1$ such that ${d_{k,c_1} \geq d_{k,c_2}}$; in other words, $D_{c_1}=d_{k,c1} \mathbbm{1}_{\{d_{k,c_1} \geq d_{k,c_2}\}}$, where $\mathbbm{1}$ is the indicator function.

Let $X_1=d_{k,c_1}, X_2=d_{k,c_2}$ with pdfs $f_{X_1}(x)$ and $f_{X_2}(y)$ as given in~\eqref{equdkc}. We can then express the pdf of $D_{c_1}$ as follows:
\begin{align*}
    f_{D_{c_1}}(x)
        &= \int_0^x \tfrac{f_{X_1|X_2}(x|y)\mathbb{P}[X_1 \geq y] f_{X_2}(y)}{\mathbb{P}[X_1 \geq X_2] }dy \\
        &= \int_0^x \tfrac{f_{X_1|X_2}(x|y)\mathbb{P}[X_1 \geq y] {\tfrac{2y}{R_\text{C}^{2}}}\!\left(\!{\tfrac{2}{\pi}}\cos^{-1}\!\left({\tfrac{y}{2R_\text{C}}} \right)\!-\!{\tfrac{y}{\pi R_\text{C}}}\!\sqrt{1\!-\!{\tfrac{y^{2}}{4R_\text{C}^{2}}}}\right)}{\mathbb{P}[\{d_{k,c_1} \geq d_{k,c_2}] }dy\\
        &= \tfrac{1}{\mathbb{P}[\{d_{k,c_1} \geq d_{k,c_2}]}\!\!\left(\!{\tfrac{2x}{R_\text{C}^{2}}}\!\left(\!{\tfrac{2}{\pi}}\!\cos^{-1}\!\left({\tfrac{x}{2R_\text{C}}} \right)\!\!-\!{\tfrac{x}{\pi R_\text{C}}}\!\sqrt{1\!-\!{\tfrac{x^{2}}{4R_\text{C}^{2}}}}\right)\right) \\
        &~~\times  \left(1+ \tfrac{2}{\pi}\left(\tfrac{x^2}{R_\text{C}^2}-1 \right)\cos^{-1}\left( \tfrac{x}{2R_\text{C}}\right) \!\!-\!{\tfrac{x}{\pi R_\text{C}}}\left(1+ \tfrac{x^2}{2R_\text{C}} \right)  \!\sqrt{1\!-\!{\tfrac{x^{2}}{4R_\text{C}^{2}}}}\right).
\end{align*}
The $\nth{n}$ moment of $X_1$ is obtained by computing $\int_0^{2R_\text{C}}x^n f_{D_{c_1}}(x)$, from which we deduce that $\mathbb{E} \left[d_{k,c_1}\right]\approx \tfrac{512R_\text{C}}{45\pi^2}$.

\textit{Remark}: When no resource allocation is applied so that all active D2D users share resources with one CUE, the first moment of the distance between two uniformly distributed points is $\mathbb{E} \left[d_{k,c_1}\right]=128 R_\text{C}/(45\pi)$~\cite{moltchanov2012distance}.

\subsection{Proof of Theorem ~\ref{Theorem_EDPPC}}\label{proof:Therem:EDPPC}
To calculate $\mathbb{E}\left[p_k^{\frac{2}{\alpha}}\right]=\mathbb{E}\left[\text{min}(U^{\frac{2}{\alpha}} d_{k,k}^{2}, V^{\frac{2}{\alpha}}  d_{0,k}^{2})\right]$, we let $A=U^{\frac{2}{\alpha}} d_{k,k}^{2}$ and $B=V^{\frac{2}{\alpha}}  d_{0,k}^{2}$. Using the Jacobian transformation \cite{grimmett2001probability}, we have $f_A(a)=\frac{1}{R_{\mathrm{D}}^2U^{\frac{2}{\alpha}}}$, $F_A(a)=\frac{a}{R_{\mathrm{D}}^2U^{\frac{2}{\alpha}}}$, $f_B(b)=\frac{1}{R_{\mathrm{C}}^2V^{\frac{2}{\alpha}}}$ and $F_B(b)=\frac{b}{R_{\mathrm{C}}^2V^{\frac{2}{\alpha}}}$. Then,
\begin{align*}
    \mathbb{E}\left[p_k^{\frac{2}{\alpha}}\right]
        &= \int_{-\infty}^{\infty}\int_{-\infty}^{\infty} \text{min}(a,b)f_{A}(a)f_{B}(b) da\: db\\
        &=  \int_{-\infty}^{\infty}\!\!\!a f_{A}(a)\left(\int_{a}^{\infty}\!\!f_{B}(b)db\right)da
            +\int_{-\infty}^{\infty}\!\!bf_{B}(b)\left( \int_{b}^{\infty}\!\!f_{A}(a)da\right)db \\
        &=  \int_{-\infty}^{\infty}\!\!\!af_{A}(a)\left(1- F_B(a)\right)da +\int_{-\infty}^{\infty}\!\!\!bf_{B}(b)\left( 1-
            F_A(b)\right)db \\
        &=  \int_{-\infty}^{\infty}\!\!\!af_{A}(a)da +\int_{-\infty}^{\infty}bf_{B}(b)db
            -\left(\int_{-\infty}^{\infty}\!\!\!af_{A}(a)F_B(a)da+\int_{-\infty}^{\infty}bf_{B}(b)F_A(b)db\right)\\
        &=  \expec{A}+\expec{B}-\expec{\max(A,B)}.
\end{align*}

\subsection{Proof of Lemma ~\ref{Lemma_EDPPC}}\label{proof:LemmaEDPPC}
Let $A$ and $B$ be two random variables with pdfs $f_A(a)$ and $f_B(b)$, and cdfs $F_A(a)$ and $F_B(b)$, respectively. Then,
\begin{align*}
    \expec{\max(A,B)}
        &=\int_{-\infty}^{\infty}\int_{-\infty}^{\infty} \text{max}(a,b)f_{A}(a)f_{B}(b) da\: db \\
        &=\int_{-\infty}^{\infty}\!a f_{A}(a)\left(\int_{-\infty}^{a}\!\!f_{B}(b)db\right)da \:+\int_{-\infty}^{\infty}b f_{B}(b)\left( \int_{-\infty}^{b}f_{A}(a)da\right)db
\end{align*}
from which the result follows.

\subsection{Proof of Corollary ~\ref{Corollary:EDPPC}}\label{Proof:COr_EDPPC}
To calculate $\expec{\max(A,B)}$, we consider two cases.

\noindent\textit{Case 1}: If $\tiny R_{\mathrm{D}}^2 U^{2/\alpha} \triangleq a'   >  b' \triangleq R_{\mathrm{C}}^2 V^{2/\alpha}$. Then, applying~\eqref{eq:expec_maxAB}, we have
\begin{align*}
    \expec{\max(A,B)}
        &= \int_{0}^{b'}\!\!\!af_{A}(a)F_B(a)da+\int_{b'}^{ a'}\!\!\!af_{A}(a)\times 1da + \int_{0}^{b'}\!\!\!b f_{B}(b)F_A(b)db\\
        &=\frac{ b'^2}{3a'}+ \frac{a'^2-b'^2}{2a'}+\frac{b'^2}{3a'}=\frac{ b'^2}{6a'}+ \frac{a'}{2} \\
        &=\frac{ R_{\mathrm{C}}^4 V^{4/\alpha}}{6R_{\mathrm{D}}^2 U^{2/\alpha}}+ \frac{R_{\mathrm{D}}^2 U^{2/\alpha}}{2}
\end{align*}

\noindent\textit{Case 2}: If $a' \leq b'$, then following the same approach, we obtain
\begin{equation*}
\expec{\max(A,B)}=\frac{ a'^2}{6b'}+ \frac{b'}{2}=\frac{ R_{\mathrm{D}}^4 U^{4/\alpha}}{6R_{\mathrm{C}}^2 V^{2/\alpha}}+ \frac{R_{\mathrm{C}}^2 V^{2/\alpha}}{2}
\end{equation*}

A similar expression for $\expec{\min(A,B)}$ can be easily derived by applying Theorem~\ref{Theorem_EDPPC}.

%
%

%
\bibliographystyle{IEEEtran}
\bibliography{IEEEabrv,power_control_references}

\end{document}